\documentclass[10pt,journal]{IEEEtran}

\ifCLASSINFOpdf

\else

\fi

\usepackage{booktabs}
\usepackage{calc}
\usepackage{xcolor}
\usepackage{dsfont}
\usepackage{url}
\usepackage{amsmath}
\usepackage{amssymb}
\usepackage{amsthm}
\usepackage{mathtools}
\usepackage{tabularx}
\usepackage{nicefrac}

\DeclareMathOperator*{\argmin}{arg\,min}
\usepackage{multirow}

\usepackage[switch]{lineno}
\usepackage[named]{algo}
\usepackage[noend]{algpseudocode}
\usepackage{algorithm}

\ifCLASSOPTIONcompsoc
\usepackage[caption=false,font=normalsize,labelfont=sf,textfont=sf]{subfig}
\else
\usepackage[caption=false,font=footnotesize]{subfig}
\fi

\usepackage{cite}

\usepackage[inline]{enumitem}

\usepackage{balance}
\usepackage{verbatim}

\theoremstyle{definition}
\newtheorem{definition}{Definition}
\newtheorem{proposition}{Proposition}
\newtheorem{theorem}{Theorem}

\newtheorem{lemma}{Lemma}
\newtheorem{remark}{Remark}

{
	\theoremstyle{plain}
	
}

\newtheorem{innercustomgeneric}{\customgenericname}
\providecommand{\customgenericname}{}

\DeclareMathOperator*{\minimize}{minimize}

\DeclareMathOperator*{\argmax}{arg\,max}

\usepackage{stackengine}
\def\delequal{\mathrel{\ensurestackMath{\stackon[1pt]{=}{\scriptstyle\Delta}}}}

\usepackage{booktabs}

\definecolor{myBlue}{RGB}{219, 48, 122}

\begin{document}
	\title{WaveMax: Radar Waveform Design via Convex Maximization of FrFT Phase Retrieval}
	
	\author{Samuel Pinilla~\IEEEmembership{Member,~IEEE}, Kumar Vijay Mishra~\IEEEmembership{Senior Member,~IEEE}, and Brian M. Sadler~\IEEEmembership{Life Fellow,~IEEE}
		
		\thanks{S. P. is with The Alan Turing Institute, London NW12DB and Rutherford Appleton Laboratory, Oxfordshire OX110QX United Kingdom (e-mail: samuel.pinilla@stfc.ac.uk).}
		\thanks{K. V. M. is with the United States DEVCOM Army Research Laboratory, Adelphi, MD 20783 USA (e-mail: kvm@ieee.org).}
		\thanks{B. M. S. is with The University of Texas, Austin, TX 78712 USA (e-mail: brian.sadler@ieee.org).}
		\thanks{K. V. M. acknowledges partial support from the National Academies of Sciences, Engineering, and Medicine via Army Research Laboratory Harry Diamond Distinguished Fellowship. Research was sponsored by the Army Research Laboratory and was accomplished under Cooperative Agreement Number W911NF-21-2-0288. The views and conclusions contained in this document are those of the authors and should not be interpreted as representing the official policies, either expressed or implied, of the Army Research Laboratory or the U.S. Government. The U.S. Government is authorized to reproduce and distribute reprints for Government purposes notwithstanding any copyright notation herein.}
		\thanks{The conference precursor of this work was presented at the 2021 IEEE International Symposium on Information Theory (ISIT).}
	}

	
	\maketitle
	
	\begin{abstract}
		The ambiguity function (AF) is a critical tool in radar waveform design, representing the two-dimensional correlation between a transmitted signal and its time-delayed, frequency-shifted version. Obtaining a radar signal to match a specified AF magnitude is a bi-variate variant of the well-known phase retrieval problem. Prior approaches to this problem were either limited to a few classes of waveforms or lacked a computable procedure to estimate the signal. Our recent work provided a framework for solving this problem for both band- and time-limited signals using non-convex optimization. In this paper, we introduce a novel approach \textit{WaveMax} that formulates waveform recovery as a \textit{convex} optimization problem by relying on the fractional Fourier transform (FrFT)-based AF. We exploit the fact that AF of the FrFT of the original signal is equivalent to a rotation of the original AF. In particular, we reconstruct the radar signal by solving a low-rank minimization problem, which approximates the waveform using the leading eigenvector of a matrix derived from the AF. Our theoretical analysis shows that unique waveform reconstruction is achievable with a sample size no more than three times the signal frequencies or time samples. Numerical experiments validate the efficacy of \textit{WaveMax} in recovering signals from noiseless and noisy AF, including scenarios with randomly and uniformly sampled sparse data.
	\end{abstract}
	
	
	\IEEEpeerreviewmaketitle
	
	\section{Introduction}
	Signal processing for radar and sonar sensing is heavily dependent on cross-correlation for signal detection and parameter estimation. A key result of this process is the \textit{ambiguity function} (AF)~\cite{jing2018designing,chen2022generalized,wang2024designing}, which is formed by correlating the transmitted waveform with its Doppler-shifted and time-delayed replicas. Originally introduced by Ville \cite{ville1948theorie}, the importance of AF as a radar signal design metric was elucidated in Woodward's groundbreaking work \cite{woodward1965probability,woodward1967radar}, with further elaboration provided by Siebert \cite{siebert1956radar}. The definition of AF is not unique \cite{abramovich2008bounds,baylis2016myths} and there are several variants to accommodate a wider range of radar target scenarios \cite{urkowitz1962generalized,san2007mimo}.  
	
	 AF plays a crucial role in selecting the appropriate radar waveforms for specific applications~\cite{wilcox1960synthesis}. Although an ideal AF would be non-zero at a single point, the likelihood that a target aligns precisely with such a response is negligible \cite{jankiraman2007design,rihaczek1996principles}. In practice, considerable theoretical effort has been dedicated to identifying AFs that deliver optimal radar performance. This involves designing waveforms that produce AFs with a \textit{thumbtack} shape, characterized by a sharp central spike and low sidelobes in the delay-Doppler plane. The inverse problem of designing a signal that yields a given AF dates back to the seminal work of Rudolf de Buda in 1970 \cite{de1970signals}. This study demonstrated that when the AF is bounded by a \textit{Hermite function}, the corresponding signal is also a Hermite function, with the Hermite polynomial coefficients determined by comparing the AF's coefficients \cite{de1970signals}. 
	
	The AF may also be viewed as a bilinear or quadratic time-frequency representation (TFR) that has been obtained after smoothing or filtering the \textit{Wigner-Ville distribution} (WVD) \cite{boashash2016time}. In general, several TFRs exist \cite[Chapter 2]{boashash2016time} and their analysis is used to characterize the time-varying spectral contents of nonstationary signals. For example, representation properties of WVD are shown to be optimal in the analysis of linearly frequency modulated (FM) signals \cite{boashash1998polynomial} but it is suboptimal for nonlinear FM (NLFM) \cite{barkat1999design,hussain2002adaptive,o2009improved}. In this context, subsequent studies \cite{jaming2010phase,jaming2014uniqueness} on AF-based radar waveform design explored signal design via other TFRs, such as using the \textit{fractional Fourier transform} (FrFT) \cite{sejdic2011fractional,ozaktas2001fractional}. In this paper, we focus on waveform design using the FrFT-based AF.
	
	First introduced by Namias in the context of quantum mechanics \cite{namias1980fractional}, the FrFT extends the conventional Fourier transform (FT) by incorporating an additional degree of freedom: rotation \cite{pei2001relations}. Various formulations of FrFT have been proposed \cite{erseghe1999unified}, including the weighted FrFT \cite{shih1995fractionalization}. Broadly, the FrFT of order $\alpha$ represents a rotation of a signal in the time-frequency plane by an angle $\alpha$. For specific values, $\alpha=2 n \pi$, $n\pi$, and $2 n \pi+\frac{\pi}{2}$, where $n$ is an integer, the FrFT simplifies to the time-domain signal, its time-reversal, and its FT, respectively \cite{sejdic2011fractional}. Numerous approaches for discrete-time FrFT (DFrFT) have also been developed \cite{ma2020fractional,candan2000discrete,zhao2021compressed}, with some lacking closed-form solutions \cite{kunche2020fractional}. In this work, we adopt a DFrFT based on the direct sampling of the FrFT \cite{tao2007sampling,jaming2014uniqueness}.
	
	In particular, the AF of signal's FrFT is equivalent to the rotation of Woodward's definition of AF for the same signal \cite{pei2001relations,almeida1994fractional}. This relationship has been exploited for a variety of radar applications, including generation of phase-coded waveforms \cite{gao2022fractional}, analysis of signals that are bandlimited in FrFT domain \cite{erseghe1999unified}, and imaging with inverse synthetic aperture radar \cite{borden2006fractional}. Leveraging this approach, \cite{jaming2010phase} extended de Buda's results by showing that the bounding assumption over the radar AF is removed to uniquely identify (up to trivial ambiguities) a Hermite function and \textit{rectangular pulse trains}. 
	
	The approaches in \cite{de1970signals, jaming2010phase} introduced the formulation of radar waveform design based on AF as a one-dimensional (1-D) bivariate phase recovery (PR) problem, which involves the recovery guarantees of complex signals from measurements only in magnitude~\cite[Page 35]{he2012waveform}. Although PR has been widely studied in fields such as optics \cite{pinilla2023unfolding}, image processing \cite{jacome2024invitation}, X-ray crystallography \cite{harrison1993phase}, speech recognition \cite{juang1993fundamentals}, and astronomy \cite{fienup1987phase}, its application to radar waveform design has received limited attention. Conventional AF-based waveform design typically relies on iterative approximation methods \cite{arikan1990time, sen2009adaptive} or optimization frameworks \cite{alhujaili2019quartic} to achieve the desired AF shape. In this context, \textit{exact} signal recovery from AF-based PR had until now only been demonstrated for specific cases, such as Hermite functions through the Woodward's AF \cite{de1970signals} and rectangular pulse trains via the FrFT-based AF \cite{jaming2010phase}. \textit{However, this body of work lacks a computable method to recover the waveform from the AF and fails to provide performance guarantees}. In addition, it is not straightforward to generalize them to waveforms that are widely used by radar systems, such as band- or time-limited pulses \cite{theron1999ultrawide,chen2008mimo}. Although results from \cite{jaming2014uniqueness} have been employed to construct low redundancy frame for stable PR \cite{bodmann2015stable}, theoretical uniqueness results and practical recovery algorithms for the original FrFT-based AF PR remain largely unexplored in the literature.
	
	\begin{figure*}[t]
		\centering
		\includegraphics[width=0.96\linewidth]{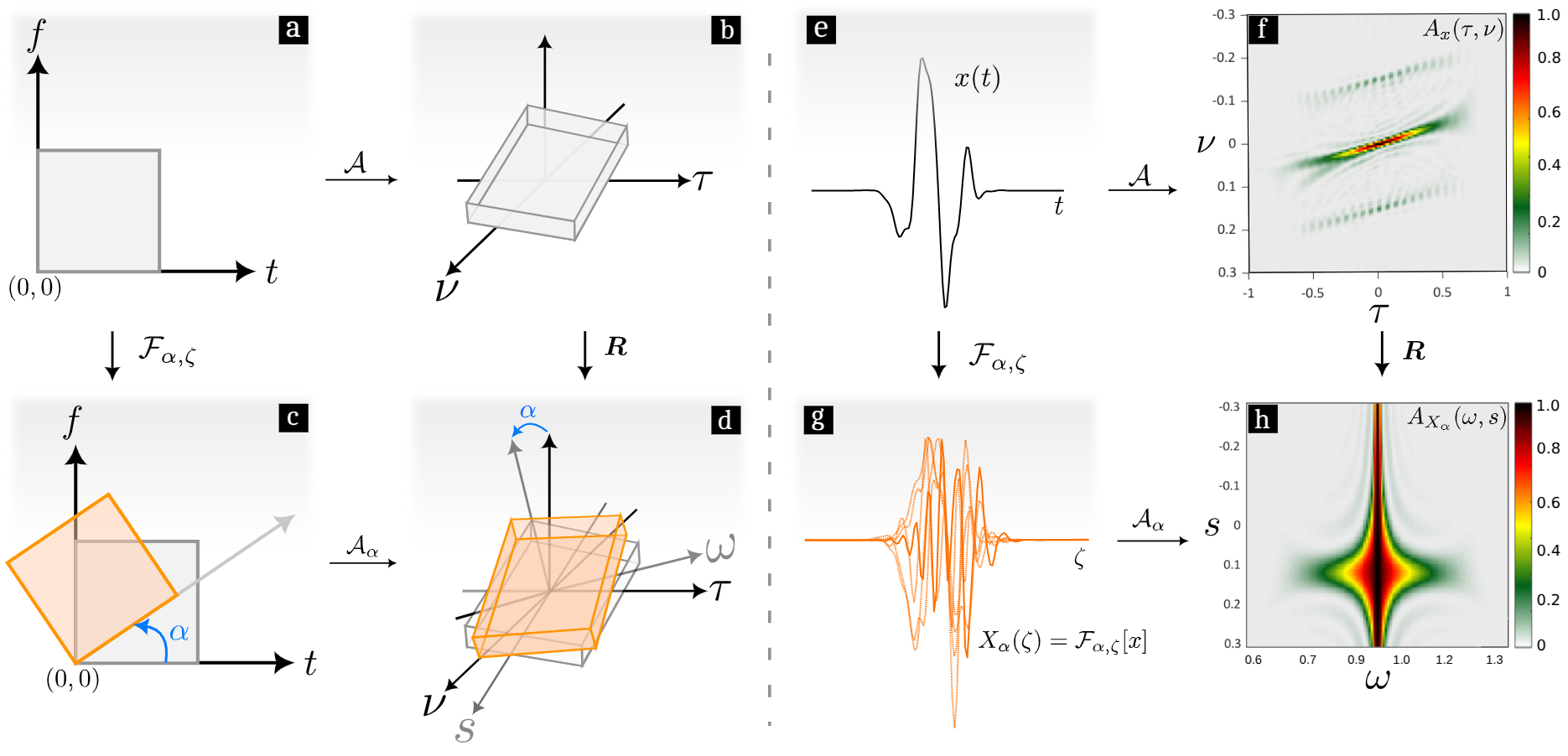}
		\vspace{-1em}
		\caption{An illustration of the relationship between the FrFT-based AF  $A_{X_{\alpha}}(\omega,s)$ and the conventional AF $A_{x}(\tau,\nu)$ of the signal $x(t)$. (a) Signal $x$ represented in the time-frequency $t$-$f$ plane. (b) Representation of the AF of the same signal in the delay-Doppler $\tau$-$\nu$ plane obtained using the operator $\mathcal{A}$ on $x(t)$. (c) Illustration of rotation of the signal in the $t$-$f$ plane by applying the FrFT operator $\mathcal{F}_{\alpha,\zeta}$. (d) Representation of the AF of the FrFT $X_{\alpha}$ shows that $A_{X_{\alpha}}$ is the rotation of $A_x$ using the coordinate transformation matrix $\boldsymbol{R}$ to the new $\omega$-$s$ plane. (e) Example of a signal $x(t)$. (f) AF $A_x(\tau,\nu)$  in the normalized delay-Doppler plane. (g) FrFT $X_{\alpha}(\zeta)$ plotted for various values of $\alpha$ from $-\pi/2$ to $\pi/2$ in increments of approximately $\pi/128$. (h) The FrFT-based AF $A_{X_{\alpha}}(\omega,s)$ is given by $A_{X_{\alpha}}(\omega,s)=A_{x}(\tau \cos \alpha-\nu \sin \alpha, \tau \sin \alpha+ \nu \cos \alpha)$.}
		\label{fig:intro}
		\vspace{-1.5em}
	\end{figure*}
	
	Our recent work \cite{pinilla2024phase} formulated the AF-based waveform design of band- and time-limited signals as a non-convex PR problem. This approach provided strong performance guarantees for waveform retrieval, and the algorithm was able to recover waveforms in challenging scenarios such as noisy and sparsely sampled AFs. In this paper, we present a convex formulation of AF-based PR to design band-limited radar signals. We achieve this improvement over the approach in \cite{pinilla2024phase} by employ the FrFT-based AF PR, which remained an open problem in prior works such as \cite{jaming2010phase,jaming2014uniqueness}. We solve this problem by developing \textit{WaveMax} - radar \textit{wav}eform recovery through convex \textit{max}imization - algorithm. This is a basis pursuit method that requires a designed approximation of the radar signal obtained by extracting the leading eigenvector of a AF-dependent matrix. Our performance analyses show that the radar signals are reconstructed using signal samples no more than three times the number of signal frequencies or time samples. Numerical experiments demonstrate that \textit{WaveMax} recovers band-limited signals in a variety of sampling (even-sparsely and randomly sampled AFs) and noisy (full noiseless samples and sparse noisy samples) scenarios.
	
	Our proposed WaveMax addresses a critical limitation in the recent formulation of AF-based waveform design for band- and time-limited signals as a non-convex phase retrieval problem \cite{pinilla2024phase}. Although \cite{pinilla2024phase} demonstrated robust numerical performance --including operation under noisy and sparsely sampled AF conditions-- it inherently precludes uniqueness guarantees because of its non-convex formulation. Furthermore, although the fundamental theory of the FrFT AF phase recovery \cite{pei2001relations,almeida1994fractional,jaming2010phase} establishes theoretical relationships, \textit{it does not provide a computationally tractable algorithm for waveform recovery}. WaveMax ensures unique waveform recovery from AF representations. Consequently, the most relevant comparison between WaveMax and the literature lies with the phase retrieval technique in \cite{pinilla2024phase}, which devised a non-convex solver but offered no uniqueness guarantees. 
	
	The preliminary results of this work appeared in our conference publication \cite{pinilla2021wavemax}, where the initialization of \textit{WaveMax} depended on the data and theoretical guarantees were excluded. In this paper, we propose a data-independent initialization procedure, prove its convergence, and provide uniqueness guarantees for the \textit{WaveMax} algorithm. In addition, we compare \textit{WaveMax} with our previous AF-based PR recovery methods, e.g. \textit{ban}d-limited \textit{ra}dar \textit{w}aveform design via phase retrieval (BanRaW) \cite{pinilla2021banraw}, and investigate additional practical constraints on signals such as LFM/NLFM waveforms. Our work addresses the inverse problem of uniquely reconstructing the waveform from a given ambiguity function (AF). In contrast, information-theoretic waveform design \cite{jing2018designing,chen2022generalized,wang2024designing} is a forward problem: it seeks to optimize waveform parameters without requiring exact reconstruction of a specific signal sequence. Therefore, these two approaches are not directly comparable. AF-based waveform design aims to minimize range and Doppler ambiguities, thereby enhancing resolution and target detection. Information-theoretic criteria, on the other hand, directly optimize waveform parameters to maximize information transfer, either by reducing estimation error through tighter bounds or by maximizing mutual information. For this reason, we compare our method with another AF-based approach (e.g., the BanRaW algorithm), since the objective of WaveMax is exact signal reconstruction, not parameter optimization.
	
	The remainder of the paper is organized as follows. In the next section, we introduce the desiderata on FrFT-based AF. In Section~\ref{sec:problem}, we formulate the radar PR problem and cast it as a convex optimization problem for the \textit{WaveMax} algorithm. Section \ref{sec:initialization} provides a mathematical description of the proposed estimation procedure depending on the AF. We validate our models and methods through numerical experiments in Section \ref{sec:results} before concluding in Section \ref{sec:conclusion}.
	
	Throughout this paper, we denote the sets of positive and strictly positive real numbers by $\mathbb{R}_{+}:=\{s\in \mathbb{R}: s\geq 0\}$ and $\mathbb{R}_{++}:=\{s\in \mathbb{R}: s>0\}$, respectively. We use boldface lowercase and uppercase letters for vectors and matrices, respectively. The sets are denoted by calligraphic letters and $\overline{\overline{ \mathcal{I} }}$ represents the cardinality of the set $\mathcal{I}$. We denote operators by calligraphic letters followed by curly braces e.g. $\mathcal{A}\{\cdot\}$. The conjugate and conjugate transpose of the vector $\boldsymbol{q}\in \mathbb{C}^{N}$ are denoted by the notation $\overline{\boldsymbol{q}}\in \mathbb{C}^{N}$ and $\boldsymbol{q}^{H}\in \mathbb{C}^{N}$, respectively. The $n$th entry of a vector $\boldsymbol{q}$, assumed to be periodic, is $\boldsymbol{q}[n]$. The $(k,l)$-th entry of a matrix $\boldsymbol{A}$ is $\boldsymbol{A}[k,l]$. We denote the FT of a vector and its conjugate reflected version (that is, $\hat{\boldsymbol{q}}[n] := \overline{\boldsymbol{q}}[-n]$) by $\tilde{\boldsymbol{q}}$ and $\hat{\boldsymbol{q}}$, respectively. The function $\lfloor \cdot \rfloor$ ($\lceil \cdot \rceil$) yields the largest (smallest) integer smaller (greater) than its argument. For matrices, define $\lVert \boldsymbol{Q}\rVert_{p}= \left[\sum_{n}\sigma_{n}^p(\boldsymbol{Q})\right]^{1/p}$ as the $p$-norm, where $\sigma_{n}(\boldsymbol{Q})$ denotes the $n$th singular value of $\boldsymbol{Q}$, the operation $\text{Tr}(\cdot)$ as the trace, and $\lVert \cdot \rVert_{\mathcal{F}}$ denotes the Frobenius norm of a matrix. For vectors, $\lVert\boldsymbol{q}\rVert_{p}$ is the usual $\ell_{p}$-norm. Additionally, we use $\sqrt{\cdot}$ is the point-wise square root; superscript within parentheses as $(\cdot)^{(t)}$ indicates the value at $t$-th iteration; $\lVert \cdot \rVert_{\mathcal{F}}$ denotes the Frobenius norm of a matrix;  $\sigma_{\max}(\cdot)$ ($\sigma_{\min}(\cdot)$) represents the largest (smallest) singular value of its matrix argument; $\mathcal{R}(\cdot)$ denotes the real part of its complex argument; $\mathbb{E}[\cdot]$ represents the expected value; and $\omega=e^{\frac{2\pi i}{n}}$ is the $n$-th root of unity, where $i=\sqrt{-1}$.
	
	\section{Desiderata for FrFT-based AF}
	\label{sec:sysmod}
	By Woodward's definition, the AF of a known narrow-band continuous-time waveform $x(t)$ is obtained by correlating the waveform with its replica shifted by Doppler frequency $\nu$ and time-delay $\tau$ as \cite{woodward1965probability}
	\begin{align}
		A_x(\tau, \nu) = \left \lvert \int\limits_{\mathbb{R}} x(t)x^{*}(t-\tau)e^{-i2 \pi \nu t} dt \right \rvert^{2}.
		\label{eq:narrow}
	\end{align}
	Denote the process of computing the AF for signal $x(t)$ by operator $\mathcal{A}\{x(t)\} = A_x(\tau, \nu)$. 
	
	Consider the FrFT\footnote{The FrFT is also derived from the eigenfunctions of the FT \cite{namias1980fractional}. This leverages the property that Hermite-Gaussian functions, which serve as eigenfunctions of the FT, remain eigenfunctions under the FrFT, with eigenvalues equal to the $\alpha$-th root of those associated with the FT \cite{namias1980fractional}.} of $x(t)$ \cite{almeida1994fractional}:
	\begin{equation}\label{eq:frft}
		X_{\alpha}(\zeta) = e^{-\mathrm{i}\pi \zeta^{2} \cot(\alpha)} \int\limits_{\mathbb{R}} x(t)e^{-\mathrm{i}\pi t^{2}\cot(\alpha)}e^{\frac{-\mathrm{i} 2\pi t\zeta}{\sin(\alpha)}}dt,
	\end{equation}
	where $\alpha\in \mathbb{R}\setminus \pi \mathbb{Z}$ is the order (rotation angle) of the FrFT, and $\zeta$ models the FrFT frequency that is related to $\nu$ in \eqref{eq:narrow} by $\nu=\frac{\zeta}{\sin(\alpha)}$. Denote the process of computing the FrFT of signal $x(t)$ by the operator $\mathcal{F}_{\alpha,\zeta}\{x(t)\} = X_{\alpha}(\zeta)$. For $\alpha=2n\pi$ and $n\pi$, $\mathcal{F}_{\alpha,\zeta}\{x(t)\}$ is defined to be $x(\zeta)$ and $x(-\zeta)$, respectively. 
	
	An important property of the FrFT, which establishes its connection with the AF is that a FrFT produces a rotation of AF (Figure \ref{fig:intro}). For instance, consider the AF of $X_{\alpha}(\zeta)$ for the time-shift $\tau=0$ and frequency $\nu$: \par\noindent\small
	\begin{align}
		A_{X_{\alpha}}(0,\nu) &=\overbrace{\left \lvert \int\limits_{\mathbb{R}} \mathcal{F}_{\alpha,\zeta}\{x(t)\}\mathcal{F}^{*}_{\alpha,\zeta}\{x(t)\}e^{-\mathrm{i}2 \pi \zeta \nu} d\zeta \right \rvert^{2}}^{A(0,\nu) \text{ as in \eqref{eq:narrow} for signal }\mathcal{F}_{\alpha,\zeta}\{x(t)\} } \nonumber\\
		&= \left \lvert  \int\limits_{\mathbb{R}} \left \lvert\int\limits_{\mathbb{R}} x(t)e^{-\mathrm{i}\pi t^{2} \cot(\alpha)-\frac{\mathrm{i}2\pi t\zeta}{\sin(\alpha)}} dt \right \rvert^{2} e^{-\mathrm{i}2\pi\zeta \nu } d\zeta \right\rvert^{2}.
		\label{eq:fractional}
	\end{align}\normalsize
	Here, $\zeta$ is the indexing Fourier frequency variable of the FrFT-based AF $A_{X_{\alpha}}$. Observe that 
	\begin{align}
		&\left \lvert\int\limits_{\mathbb{R}} x(t)e^{-\mathrm{i}\pi t^{2} \cot(\alpha)-\frac{\mathrm{i}2\pi t\zeta}{\sin(\alpha)}} dt \right \rvert^{2} \nonumber\\
		&=\iint\limits_{\mathbb{R}}x(t_{1})x^{*}(t_{2})e^{-i \pi \cot(\alpha)(t_{1}^{2}-t_{2}^{2})} e^{-i2\pi \zeta\left(\frac{t_{1}}{\sin(\alpha)}-\frac{t_{2}}{\sin(\alpha)}+\nu\right)} dt_{1}dt_{2}
	\end{align}
	Using the above equivalence and \eqref{eq:fractional} gives
	\begin{align}
		&A_{X_{\alpha}}(0,\nu)\nonumber\\ 
		&= \left \lvert\iiint\limits_{\mathbb{R}}x(t_{1})x^{*}(t_{2})e^{-i \pi \cot(\alpha)(t_{1}^{2}-t_{2}^{2})}\right.\nonumber\\
		&\hspace{1cm}\left. e^{-i2\pi \zeta\left(\frac{t_{1}}{\sin(\alpha)}-\frac{t_{2}}{\sin(\alpha)}+\nu\right)} dt_{1}dt_{2}d\zeta\right\rvert^{2} \nonumber\\
		&=\left \lvert\iint\limits_{\mathbb{R}}x(t_{1})x^{*}(t_{2})e^{-i \pi \cot(\alpha)(t_{1}^{2}-t_{2}^{2})}dt_{1}dt_{2} \right. \nonumber\\
		&\hspace{1cm}\left. \left(\int\limits_{\mathbb{R}}e^{-i2\pi \zeta\left(\frac{t_{1}}{\sin(\alpha)}-\frac{t_{2}}{\sin(\alpha)}+\nu\right)} d\zeta\right)\right\rvert^{2}.
		\label{eq:FrFTEquivalency}
	\end{align}
	Observe that the last integral with respect to $\zeta$ in \eqref{eq:FrFTEquivalency} vanishes, except when $\frac{t_{1}}{\sin(\alpha)}-\frac{t_{2}}{\sin(\alpha)}+\nu=0$ leading to $t_{2}=t_{1} + \nu\sin(\alpha)$. Hence, \eqref{eq:FrFTEquivalency} is equivalent to
	\begin{small}
		\begin{align} 
			&A_{X_{\alpha}}(0,\nu) \nonumber\\
			&= \left \lvert\int\limits_{\mathbb{R}}x(t)x^{*}(t+\nu\sin(\alpha))  e^{-i \pi \cot(\alpha)(t^{2}-t^{2}-2t\nu\sin(\alpha)-\nu^{2}\sin^{2}(\alpha))} dt\right\rvert^{2} \nonumber\\
			&= \left \lvert\int\limits_{\mathbb{R}}x(t)x^{*}(t+\nu\sin(\alpha)) e^{i 2\pi t \nu \cos(\alpha)} dt\right\rvert^{2}.
			\label{eq:FrFTEquivalency1}
		\end{align}
	\end{small}
	Note that \eqref{eq:FrFTEquivalency1} being computed over $\mathbb{R}$, with the change of sign in the exponential, it 
	is equivalent to 
	\begin{align}
		A_{X_{\alpha}}(0,\nu) &= \left \lvert\int\limits_{\mathbb{R}}x(t)x^{*}(t+\nu\sin(\alpha)) e^{-i 2\pi t \nu \cos(\alpha)} dt\right\rvert^{2} \nonumber\\
		&=A_{x}(-\nu \sin(\alpha),\nu\cos(\alpha)).
		\label{eq:FrFTEquivalency2}
	\end{align}
	
	It follows from \eqref{eq:FrFTEquivalency2} that a fractional FT produces a rotation of AF in the $(\tau,\nu)$ plane $(\omega,s)$~\cite{sejdic2011fractional} (Figure~\ref{fig:intro}). The illustration below shows these relationships:
	\par\noindent\small
	\begin{align*}
		\begin{array}{lll}
			x(t) & \xrightarrow[]{\mathcal{A}} &  A_{x}(\tau, \nu) \\
			\downarrow \text { FrFT } &  & \downarrow \text {Rotation of AF } \\
			X_{\alpha}(\zeta)=\mathcal{F}_{\alpha,\zeta}\{x(t)\} & \xrightarrow[]{\mathcal{A}_{\alpha}} 
			& A_{X_{\alpha}}(0, \nu)=A_{x}(-\nu\sin(\alpha),\nu\cos(\alpha)) .
		\end{array}
	\end{align*}\normalsize
	
	In general, for the cross-correlation AF \cite{jaming2010phase}, we have $A_{X_{\alpha}}(\omega,s)=A_{x}(\tau,\nu)$, where the coordinates $(\omega,s)$ in the rotated frame are related to $(\tau,\nu)$ as~\cite{jaming2010phase}
	\begin{align}
		\left[\begin{array}{l}
			\omega \\
			s
		\end{array}\right]=\underbrace{\left[\begin{array}{rr}
				\cos \alpha & -\sin \alpha \\
				\sin \alpha & \cos \alpha
			\end{array}\right]}_{\boldsymbol{=R}}\left[\begin{array}{l}
			\tau \\
			\nu
		\end{array}\right],
	\end{align}
	where $\boldsymbol{R}$ is the coordinate rotation transformation matrix. In case of auto-coorelation AF in \eqref{eq:narrow}, we set $\tau=0$ above to obtain the relationship between AF and FrFT AF.
	
	\section{Problem Formulation}
	\label{sec:problem}
	In practice, it is difficult to compute FrFT through numerical integrations. Therefore, its discrete-time formulation is utilized instead. Assume the sampling periods in time and Doppler domains are $\Delta t$ and $\Delta \nu$, respectively, such that the discrete indices are $p_{\alpha}=\nu\sin(\alpha)/\Delta t$ and $k=f/\Delta \nu$. Given a discrete-time signal $\boldsymbol{x}\in \mathbb{C}^{N}$, the discrete-valued version of its FrFT-based AF in \eqref{eq:fractional} is \cite{jaming2010phase}
	\begin{equation}
		\boldsymbol{A}[\alpha,k]:= \left\lvert \sum_{n=0}^{N-1} \boldsymbol{x}[n]\overline{\boldsymbol{x}}[n+p_{\alpha}]e^{\frac{-2i\pi nk\cos(\alpha)}{N}} \right\rvert^2,
		\label{eq:Ambiguity}
	\end{equation}
	where $i = \sqrt{-1}$, $k=0,\cdots,N-1$, and $\alpha\in [-\pi/2,\pi/2]$ uniformly sampled. The AF in \eqref{eq:Ambiguity} is a map $\mathbb{C}^{N}\rightarrow \mathbb{R}_{+}^{N\times N}$ that has four types of symmetry or \textit{trivial ambiguities} as follows
	\begin{description}
		\item[T1] the rotated signal $e^{i\phi}\boldsymbol{x}[n]$ for some $\phi\in \mathbb{R}$.
		\item[T2] the translate signal $\boldsymbol{x}[n-a]$ for some $a\in \mathbb{R}$.
		\item[T3] the reflected signal $\boldsymbol{x}[-n]$.
		\item[T4] the scaled signal $e^{ibn}\boldsymbol{x}[n]$ for some $b\in \mathbb{R}$.
		\label{eq:ambiguities}
	\end{description}
	Our goal is to estimate the signal $\boldsymbol{x}$, up to these trivial ambiguities, from the AF $\boldsymbol{A}$. From the definition of $\boldsymbol{A}$, we also define $m$ as the number of measurements which corresponds to the number of available entries of $\boldsymbol{A}$. 
	
	\subsection{Uniqueness guarantees}
	We introduce the following definition of a band-limited signal.
	\begin{definition}[$B$-band-limitedness]
		A signal $\boldsymbol{x}\in \mathbb{C}^{N}$ is defined to be $B$-band-limited if its FT $\tilde{\boldsymbol{x}}\in \mathbb{C}^{N}$ contains $N-B$ consecutive zeros. That is, there exists $k$ such that $\tilde{\boldsymbol{x}}[k]=\cdots=\tilde{\boldsymbol{x}}[N+k-B-1]=0$.
		\label{def:bandlimitedSignal}
	\end{definition}
	Following this definition, we show in the following Proposition~\ref{prop:uniqueness} that the AF is able to uniquely identify a band-limited signal. Here, \textit{almost all} implies that the set of signals, which is not uniquely determined up to trivial ambiguities, is contained in the vanishing locus of a nonzero polynomial.
	\begin{proposition}
		Assume $\boldsymbol{x}\in \mathbb{C}^{N}$ is a $B$-band-limited signal for some $B\leq N/2$. Then, almost all signals are uniquely determined from their AF $\boldsymbol{A}[p,k]$, up to trivial ambiguities, from $m\geq 3B$ measurements. If, in addition, the spectrum signal $|\tilde{\boldsymbol{x}}[t]|^2$ is also known and $N \geq 3$, then $m\geq2B$ measurements suffice.
		\label{prop:uniqueness}
	\end{proposition}
	\begin{IEEEproof}
		Consider the expression of AF $\boldsymbol{A}$, i.e., \cite{jaming2010phase}
		\begin{align}
			\boldsymbol{A}[\alpha,k] = \left\lvert \sum_{n=0}^{N-1} \boldsymbol{x}[n]\overline{\boldsymbol{x}}[n+p_{\alpha}] e^{\frac{-2\pi in k\cos(\alpha)}{N}} \right\rvert^{2},
			\label{eq:fractionalEqui}
		\end{align}
		for $k=0,\cdots,N-1$, and $\alpha\in [-\pi/2,\pi/2]$. 
		
		Assume $B = N/2$, $N$ is even, that $\tilde{\boldsymbol{x}}[n] \not= 0$ for $n = 0\dots,B-1$, and that $\tilde{\boldsymbol{x}}[n] = 0$ for $n = N/2,\dots, N-1$. If the signal's nonzero coefficients are not in the interval $0,\dots, N/2-1$, then we cyclically reindex the signal without affecting the proof. If $N$ is odd, then replace $N/2$ by $\lfloor N/2 \rfloor$ everywhere in the sequel. Clearly, the proof carries through for any $B \leq N/2$. 
		
		Performing an analogous construction procedure as in \cite[Theorem 1]{pinilla2024phase} over \eqref{eq:fractionalEqui}, almost all signals $\boldsymbol{x}$ are uniquely determined from $m\geq 3B$ measurements. When the signal spectrum $\lvert \tilde{\boldsymbol{x}} \rvert$ is also provided, $\tilde{\boldsymbol{x}}$ is uniquely determined for $N\geq 3$. This implies that, under this scenario, only $m\geq 2B$ measurements are needed.
	\end{IEEEproof}
	
	Similar to \cite{pinilla2024phase}, Proposition~\ref{prop:uniqueness} is easily extended to time-limited signals. Later, in Section \ref{sec:results}, we numerically validate both band- and time-limited cases. We note that the signals $\boldsymbol{x}[t]$, $\boldsymbol{x}[t]e^{\textrm{i}\phi}$, $\boldsymbol{x}[t-a]$, $\boldsymbol{x}[-t]$, and $e^{\textrm{i}bt}\boldsymbol{x}[t]$ yield the same magnitude measurements for any constant $\phi$, $a$, $b \in \mathbb{R}$ and therefore these constants cannot be recovered by any algorithm. This naturally leads to the following measure of the relative error between the true signal $\boldsymbol{x}$ and any $\boldsymbol{q}\in \mathbb{C}^{N}$.
	
	\begin{definition}
		\label{def:dis}The distance between the two vectors $\boldsymbol{x} \in \mathbb{C}^N$ and $\boldsymbol{q} \in \mathbb{C}^N$ is defined as 
		\begin{equation}
			\text{dist}(\boldsymbol{x},\boldsymbol{q}):= \min_{\boldsymbol{z}\in \mathcal{T}(\boldsymbol{x})}\lVert \boldsymbol{q}-\boldsymbol{z}\rVert_{2},
			\label{eq:distance}
		\end{equation}
		where the set $\mathcal{T}(\boldsymbol{x})=\{\boldsymbol{z}\in \mathbb{C}^{N} \hspace{0.2em}:\hspace{0.2em} \boldsymbol{z}[n]=e^{i\beta}e^{i b n}\boldsymbol{x}[\epsilon n - a] \text{ for } \beta,b\in \mathbb{R}, \text{ and }\epsilon=\pm 1, a\in \mathbb{Z} \}$ contains vectors with all possible trivial ambiguities. \textcolor{blue}If $\text{dist}(\boldsymbol{x},\boldsymbol{q})=0$ and the uniqueness conditions of Proposition~\ref{prop:uniqueness} are met, then, for almost all signals, $\boldsymbol{x}$ and $\boldsymbol{q}$ are equal \emph{up to trivial ambiguities}. Furthermore, when $\text{dist}(\boldsymbol{x},\boldsymbol{q})=0$, it also implies that the AF for both $\boldsymbol{x}$, and $\boldsymbol{q}$ are equal.
	\end{definition}
	The minimum of $\lVert \boldsymbol{q}-\boldsymbol{z}\rVert_{2}$ in \eqref{eq:distance} exists because $\mathcal{T}(\boldsymbol{x})$ is a closed set. We prove this in Appendix~\ref{app:distance}.
	
	\subsection{WaveMax: Convex formulation for radar waveform design}
	Evidently, it follows from Proposition~\ref{prop:uniqueness} that not all the delay steps are needed to recover the signal. Therefore, a method that works for this truncated signal is also desired. We now formulate a convex optimization problem that is able to estimate a band-limited pulse from its FrFT-based AF. Define the vector $\boldsymbol{u}_{n,\alpha}$ as the $n$th row of the orthogonal FrFT matrix transform\footnote{The FrFT matrix transformation can be computed following~\cite{bultheel2004computation}.} for a given $\alpha$, and $w = e^{i\frac{2\pi}{N}} $. Then, using the equivalence \eqref{eq:FrFTEquivalency2}, we have that the AF in \eqref{eq:Ambiguity} can be expressed as 
	\begin{align}
		\boldsymbol{A}[\alpha,k] &= \left\lvert \sum_{n=0}^{N-1} \left\lvert \boldsymbol{u}_{n,\alpha}^{H}\boldsymbol{x} \right\rvert^{2}w^{-k n} \right\rvert^2 \nonumber\\
		&=\left\lvert \sum_{n=0}^{N-1} \text{Tr}(\boldsymbol{u}_{n,\alpha}\boldsymbol{u}_{n,\alpha}^{H}\boldsymbol{x}\boldsymbol{x}^{H})w^{-k n} \right\rvert^2, \nonumber\\
		&=\left\lvert  \text{Tr}\left(\left(\sum_{n=0}^{N-1}\boldsymbol{U}_{n,\alpha}w^{-k n}\right)\boldsymbol{X}\right) \right\rvert^2, \nonumber\\
		&=\left\lvert  \text{Tr}\left(\boldsymbol{B}_{\alpha,k}\boldsymbol{X}\right) \right\rvert^2, \nonumber\\
		\label{eq:system1}
	\end{align}
	where matrix $\boldsymbol{B}_{\alpha,k}$ is defined as
	\begin{align}
		\boldsymbol{B}_{\alpha,k} = \sum_{n=0}^{N-1}\boldsymbol{U}_{n,\alpha}w^{-k n} = \sum_{n=0}^{N-1}\boldsymbol{u}_{n,\alpha}\boldsymbol{u}_{n,\alpha}^{H}w^{-k n},
		\label{eq:matricesB}
	\end{align}
	and $\boldsymbol{X} = \boldsymbol{x}\boldsymbol{x}^{H}$.  Thus, considering \eqref{eq:system1} we formulate \textit{WaveMax}, a convex optimization problem (we prove this in the Lemma below) to estimate the pulse $\boldsymbol{x}$ by solving
	\begin{align}
		\minimize_{\boldsymbol{Z}\in \mathbb{C}^{N\times N}} \hspace{0.5em} & -\text{Tr}\left(\boldsymbol{Z}^{H}\boldsymbol{X}^{(0)} \right) \nonumber\\
		\text{ subject to }\hspace{0.5em} & \left\lvert  \text{Tr}\left(\boldsymbol{B}_{\alpha,k}\boldsymbol{Z}\right) \right\rvert \leq \sqrt{\boldsymbol{A}[\alpha,k]} \nonumber\\
		& \boldsymbol{Z} \succeq 0,
		\label{eq:wavemax}
	\end{align}
	where $\boldsymbol{X}^{(0)}\in \mathbb{C}^{N\times N}$, given by $\boldsymbol{X}^{(0)} = \boldsymbol{x}^{(0)}(\boldsymbol{x}^{(0)})^{H}$ with $\boldsymbol{x}^{(0)}\in \mathbb{C}^{N}$, is an approximated version of $\boldsymbol{x}\boldsymbol{x}^{H}$.
	
	The main idea behind of \eqref{eq:wavemax} is to find the positive semidefinite matrix $\boldsymbol{Z}$ that is most aligned with the approximation $\boldsymbol{X}^{(0)}$ and satisfies the relaxation of the measurement constraints in \eqref{eq:wavemax}. The construction of the estimation of $\boldsymbol{X}^{(0)}$ is explained in detail in the next section, where $\boldsymbol{x}^{(0)}$ is obtained by extracting the leading eigenvector of a matrix depending on the AF. And the following lemma states that optimization problem in \eqref{eq:wavemax} is convex.
	\begin{lemma}
		The optimization problem in \eqref{eq:wavemax} is convex.
		\label{lem:convexWavemax}
	\end{lemma}
	\begin{IEEEproof}
		See Appendix \ref{app:convexWavemax}. The key to proof lies in demonstrating that the relaxation of the measurement constraints in \eqref{eq:wavemax} is convex.
	\end{IEEEproof}

	\section{Construction of $\boldsymbol{X}^{(0)}$}
	\label{sec:initialization}
	We now develop the procedure for computing the approximated matrix $\boldsymbol{X}^{(0)}= \boldsymbol{x}^{(0)}(\boldsymbol{x}^{(0)})^{H}$. Instead of directly dealing with the AF in \eqref{eq:Ambiguity}, consider the acquired data in a transformed domain by taking its 1D IDFT with respect to the frequency variable (normalized by $1/N$) as 
	\begin{align}
		\boldsymbol{Y}[\alpha,\ell] &= \frac{1}{N}\sum_{k=0}^{N-1}\boldsymbol{A}[\alpha,k]\omega^{k \ell} \nonumber\\
		&= \frac{1}{N}\sum_{k=0}^{N-1} \left\lvert \sum_{n=0}^{N-1} \left\lvert \boldsymbol{u}_{n,\alpha}^{H}\boldsymbol{x} \right\rvert^{2}w^{-k n} \right\rvert^2 \omega^{k \ell}\nonumber\\
		&=\frac{1}{N}\sum_{k,n_{1},n_{2}=0}^{N-1} \lvert \boldsymbol{u}_{n_{1},\alpha}^{H} \boldsymbol{x} \rvert^{2} \lvert \boldsymbol{u}_{n_{2},\alpha}^{H} \boldsymbol{x} \rvert^{2} \omega^{k(\ell + n_{2}-n_{1})}, \nonumber\\
		&=\sum_{n=0}^{N-1} \lvert \boldsymbol{u}_{n+\ell,\alpha}^{H} \boldsymbol{x} \rvert^{2} \lvert \boldsymbol{u}_{n,\alpha}^{H} \boldsymbol{x} \rvert^{2},
		\label{eq:system2}
	\end{align}
	where $\ell=0,\cdots,N-1$. Then, we compute $\boldsymbol{x}^{(0)}$ as the leading eigenvector of a AF-dependent matrix using the data $\boldsymbol{Y}$, as explained in next section. We finalize with a summary of the optimization steps to estimate $\boldsymbol{x}$ in Algorithm \ref{alg:algorithm}.
	
	\subsection{Approximated vector $\boldsymbol{x}^{(0)}$}
	\label{sec:approxinit}
	Here we present a motivating example that reveals the fundamental characteristics of high-dimensional $\boldsymbol{u}_{n,\alpha}$ vectors, and band-limited signals $\boldsymbol{x}$. We fix $\boldsymbol{x}\in \mathbb{C}^{N}$ as a $\left\lceil \frac{N-1}{2} \right\rceil$-band-limited signal that conform to a Gaussian power spectrum centered at 800 nm, produced via the FT of a complex vector with a Gaussian-shaped amplitude with a cutoff frequency of $150$ microseconds$^{-1}$ (usec$^{-1}$). Then, we generate data as in \eqref{eq:system2} using the sampling vectors $\boldsymbol{u}_{n,\alpha}$. In Figure \ref{fig:init1} the ordered values of $\boldsymbol{Y}[\alpha,\ell]$ are plotted, for $N=128$ and the variable $\alpha$ (angle of FrFT) was uniformly sampled from $[\frac{\pi}{2},\frac{3\pi}{2}]$ with number of samples varying from $\lceil N/64 \rceil$ to $\lceil N/19 \rceil$. The operation $\lceil s \rceil$ returns the smallest integer greater than $s$. Observe that all values of $\boldsymbol{Y}[\alpha,\ell]$ are smaller than $10^{-1}$, which implies that $\boldsymbol{x}$ is nearly orthogonal to a large number of $\boldsymbol{u}_{n,\alpha}$.
	
	Following the Figure \ref{fig:init1} and expression of $\boldsymbol{Y}[\alpha,\ell]$ in \eqref{eq:system2}, define
	\begin{equation}
		\cos^{2}(\theta_{n,\alpha})=\frac{\lvert\langle\boldsymbol{u}_{n,\alpha},\boldsymbol{x}\rangle\rvert^{2}}{\lVert \boldsymbol{u}_{n,\alpha}\rVert^{2}\lVert \boldsymbol{x}\rVert^{2}} = \frac{\lvert\langle\boldsymbol{u}_{n,\alpha},\boldsymbol{x}\rangle\rvert^{2}}{N\lVert \boldsymbol{x}\rVert^{2}},
	\end{equation}
	where $n=0,\cdots,N-1, \alpha\in [\frac{\pi}{2},\frac{3\pi}{2}]$, and $\theta_{n,\alpha}$ is the angle between vectors $\boldsymbol{u}_{n,\alpha}$ and $\boldsymbol{x}$. Consider ordering all $\cos^{2}(\theta_{n,\alpha})$ in an ascending order, such that $\cos^{2}(\theta_{n,\alpha})\geq\cdots\geq \cos^{2}(\theta_{n,\alpha})$. In order to approximate $\boldsymbol{x}$ by a vector that is mostly orthogonal to a subset of vectors $\lbrace\boldsymbol{u}_{n,\alpha}\rbrace$ (from Figure \ref{fig:init1}), we introduce $\mathcal{I}_{0}$, an index set with cardinality $\overline{\overline{\mathcal{I}_{0}}} < N^{2}$, that includes indices of the smallest squared normalized inner-products $cos^{2}(\theta_{n,\alpha})$, where $(n,\alpha)\in \mathcal{I}_{0}$ is the collection of indices corresponding to the smallest values of $\{ \boldsymbol{Y}[\alpha,\ell] / N \}$. Thus, we approximate $\boldsymbol{x}$ by solving the following optimization problem 
	\begin{equation}
		\boldsymbol{x}^{(0)}=\argmin_{\lVert \boldsymbol{s} \rVert_{2}=1} \hspace{0.5em} \boldsymbol{s}^{H}\left(\frac{1}{N \overline{\overline{\mathcal{I}_{0}}} }\sum_{(n,\alpha)\in \mathcal{I}_{0}}\boldsymbol{u}_{n,\alpha}\boldsymbol{u}_{n,\alpha}^{H}\right)\boldsymbol{s}.
		\label{eq:coorthoIniti}
	\end{equation}
	
	\begin{figure}[t!]
		\centering
		\includegraphics[width=1\linewidth]{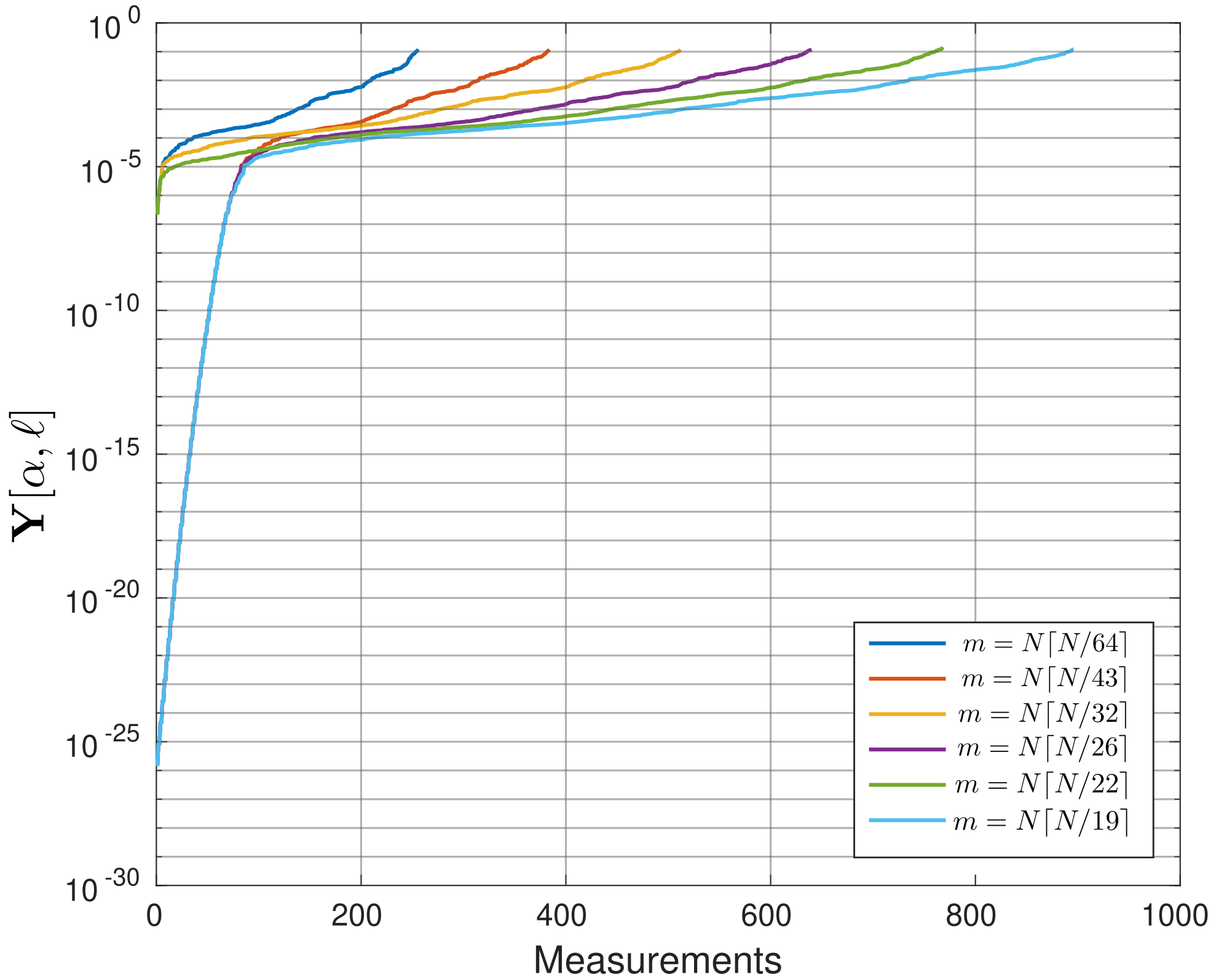}
		\caption{Ordered values of $\boldsymbol{Y}[\alpha,\ell]$ as defined in \eqref{eq:system2} for $N=128$ in log-scale. The variable $\alpha$ (angle of FrFT) was uniformly sampled from $[\frac{\pi}{2},\frac{3\pi}{2}]$ with number of samples varying from $\lceil N/64 \rceil$ to $\lceil N/19 \rceil$. Observe that all values of $\boldsymbol{Y}[\alpha,\ell]$ are smaller than $10^{-1}$, which implies that $\boldsymbol{x}$ is nearly orthogonal to a large number of $\boldsymbol{u}_{n,\alpha}$.
		}
		\label{fig:init1}
	\end{figure}
	
	Note that \eqref{eq:coorthoIniti} implies finding the smallest eigenvalue, which calls for eigen-decomposition or matrix inversion, each typically requiring computational complexity $\mathcal{O}(N^{3})$. We avoid this step by computing the summation in \eqref{eq:coorthoIniti} as 
	\begin{align}
		\frac{1}{N}\sum_{(n,\alpha)\in \mathcal{I}_{0}}\boldsymbol{u}_{n,\alpha}\boldsymbol{u}_{n,\alpha}^{H} &= \frac{1}{N}\sum_{n,\alpha=0}^{N-1}\boldsymbol{u}_{n,\alpha}\boldsymbol{u}_{n,\alpha}^{H} \nonumber\\ 
		&- \frac{1}{N}\sum_{(n,\alpha)\in \mathcal{I}^{c}_{0}}\boldsymbol{u}_{n,\alpha}\boldsymbol{u}_{n,\alpha}^{H} \nonumber\\
		&=\boldsymbol{I}_{N\times N} - \frac{1}{N}\sum_{(n,\alpha)\in \mathcal{I}^{c}_{0}}\boldsymbol{u}_{n,\alpha}\boldsymbol{u}_{n,\alpha}^{H},
		\label{eq:14}
	\end{align}
	where $\boldsymbol{I}_{N\times N}$ is an identity matrix and the orthogonality of the FrFT implies that $\mathcal{I}^{c}_{0}$ is the complement of $\mathcal{I}_{0}$. Considering the observation in \eqref{eq:14} we have that \eqref{eq:coorthoIniti} can be approximated as 
	\begin{align}
		\boldsymbol{x}^{(0)}=\argmax_{\lVert \boldsymbol{s} \rVert_{2}=1}\hspace{0.5em} \boldsymbol{s}^{H}\left(\frac{1}{N \overline{\overline{\mathcal{I}^{c}_{0}}}}\sum_{(n,\alpha)\in \mathcal{I}^{c}_{0}}\boldsymbol{u}_{n,\alpha}\boldsymbol{u}_{n,\alpha}^{H}\right)\boldsymbol{s},
		\label{eq:finalInit}
	\end{align}
	which meets the numerical computation of the leading eigenvector of matrix $\displaystyle \sum_{(n,\alpha)\in \mathcal{I}^{c}_{0}}\boldsymbol{u}_{n,\alpha}\boldsymbol{u}_{n,\alpha}^{H}$. 
	\begin{remark}
		The surrogate optimization in \eqref{eq:finalInit} -- a consequence of the FrFT orthogonal property -- is computationally less expensive than \eqref{eq:coorthoIniti}. Specifically, \eqref{eq:finalInit} is numerically solved via the power iteration method \cite{guerrero2020phase,wang2018phase}, which comprises recursively performing a matrix-vector multiplication. The significance of this result lies in enabling the initialization procedure necessary for addressing the radar PR problem, as outlined in \eqref{eq:wavemax}, to be executed using a strategy like the power iteration method \cite{pinilla2024phase}.
	\end{remark}
	
	It follows from \eqref{eq:finalInit} that $\boldsymbol{x}^{(0)}$ is unitary,  indicating that the true norm value $\lVert \boldsymbol{x} \rVert_{2}$ is lost. Therefore, we approximate $\lVert \boldsymbol{x} \rVert_{2}$ as
	\begin{align}
		\sum_{\ell=0}^{N-1}\boldsymbol{Y}[\alpha,\ell] &= \sum_{\ell=0}^{N-1}\sum_{n=0}^{N-1} \lvert \boldsymbol{u}_{n+\ell,\alpha}^{H} \boldsymbol{x} \rvert^{2} \lvert \boldsymbol{u}_{n,\alpha}^{H} \boldsymbol{x} \rvert^{2} \nonumber\\
		&=\sum_{n=0}^{N-1} \lvert \boldsymbol{u}_{n,\alpha}^{H} \boldsymbol{x} \rvert^{2}\left(\sum_{\ell=0}^{N-1}\lvert \boldsymbol{u}_{n+\ell,\alpha}^{H}\boldsymbol{x} \rvert^{2}\right), \nonumber\\
		&\cong\lVert \boldsymbol{x} \rVert_{2}^{4},
		\label{eq:auxAF}
	\end{align}
	where the second equality follows from the Parseval's theorem. 
	This leads to 
	\begin{align}
		\lVert \boldsymbol{x} \rVert_{2} \cong \lambda_{0}:= \frac{1}{N}\sum_{\alpha}\sqrt[4]{\sum_{\ell=0}^{N-1}\boldsymbol{Y}[\alpha,\ell]}.
		\label{eq:approxnorm}
	\end{align}
	
	To solve \eqref{eq:finalInit}, we employ the following Algorithm \ref{alg:initialization} for which the input is the AF $\boldsymbol{A}[\alpha,k]$ as in \eqref{eq:Ambiguity}. In Lines 2, 3 and 5, Algorithm \ref{alg:initialization} computes a random vector $\tilde{\boldsymbol{x}}^{(0)}$, and the auxiliary matrices $\boldsymbol{Y}[\alpha,\ell]$, $\boldsymbol{G}_{0}$ following \eqref{eq:system2} and \eqref{eq:finalInit} respectively. Then, in Lines 7 and 8 the recursive matrix-vector multiplication between $\boldsymbol{G}_{0}$, and $\tilde{\boldsymbol{x}}^{(0)}$ (power iteration method) to approximate the signal $\boldsymbol{x}$. Once this loop is finished, the vector $\tilde{\boldsymbol{x}}^{(T)}$ is scaled so that its norm matches approximately that of $\boldsymbol{x}$ based on \eqref{eq:approxnorm}, where $\lambda_{0} \approx \lVert \boldsymbol{x}\rVert_{2}$. 
	
	\begin{algorithm}[H]
		\caption{Spectral Initialization}
		\label{alg:initialization}
		\begin{algorithmic}[1]
			\State{\textbf{input: }Data $\left\lbrace\boldsymbol{A}[\alpha,k]:k=0,\cdots,N-1, \alpha\in [-\pi/2,\pi/2] \right\rbrace$, and maximum number of iterations $T$.}
			\State{$\tilde{\boldsymbol{x}}^{(0)}$ chosen randomly.}
			\State{Compute $$ \boldsymbol{Y}[\alpha,\ell] = \frac{1}{N}\sum_{k=0}^{N-1}\boldsymbol{A}[\alpha,k]\omega^{k \ell} $$}
			\State{\textbf{set }$\mathcal{I}^{c}_{0}$ as the set of indices corresponding to the $\lfloor m/6\rfloor$ largest values of $\{ \boldsymbol{Y}[\alpha,\ell] / N \}$.}		
			\State{Compute \begin{equation*}
					\boldsymbol{G}_{0} = \frac{1}{ \overline{\overline{\mathcal{I}^{c}_{0}}}} \sum_{(n,\alpha) \in \mathcal{I}^{c}_{0}} \boldsymbol{u}_{n,\alpha}\boldsymbol{u}_{n,\alpha}^{H}
			\end{equation*}}
			\For{$t=0:T-1$} 
			\State{$\grave{\boldsymbol{x}}^{(t+1)}= \boldsymbol{G}_{0}\tilde{\boldsymbol{x}}^{(t)}$}
			\State{$\tilde{\boldsymbol{x}}^{(t+1)}=\frac{\grave{\boldsymbol{x}}^{(t+1)}}{ \lVert \grave{\boldsymbol{x}}^{(t+1)} \rVert_{2}}$}	
			\EndFor
			\State{Compute $\boldsymbol{x}^{(0)} = \lambda_{0}\tilde{\boldsymbol{x}}^{(T)}$ where $\lambda_{0} = \frac{1}{N}\sum_{\alpha}\sqrt[4]{\sum_{\ell=0}^{N-1}\boldsymbol{Y}[\alpha,\ell]}$}		
			\State{\textbf{return: } $\boldsymbol{x}^{(0)}$ \Comment{Approximation of $\boldsymbol{x}$}}
		\end{algorithmic}
	\end{algorithm}
	
	Further, the solution to \eqref{eq:finalInit} is accurately approximated using a few power iterations,  which has a significantly lower computational complexity compared to the $\mathcal{O}(N^3)$ needed to solve \eqref{eq:coorthoIniti}. This is because it involves estimating the leading eigenvector of the matrix $\boldsymbol{G}_{0}$ \cite{guerrero2020phase}. The following Theorem \ref{theo:initi} indicates the accuracy of $\boldsymbol{x}^{(0)}$ in \eqref{eq:finalInit} to approximate the signal $\boldsymbol{x}$. It shows that \eqref{eq:finalInit} provides a close estimation of $\boldsymbol{x}$ from $\boldsymbol{Y}[\alpha,\ell]$ for almost all signals.
	\begin{theorem}
		Consider the AF $\boldsymbol{A}[\alpha,k]$ as defined in \eqref{eq:Ambiguity} for a $B$-band-limited signal $\boldsymbol{x}\in \mathbb{C}^{N}$ with $B\leq N/2$. Then, for almost all signals, the vector $\boldsymbol{x}^{(0)}\in \mathbb{C}^{N}$ returned by \eqref{eq:finalInit} satisfies
		\begin{align}
			\lVert \boldsymbol{x}\boldsymbol{x}^{H}-\boldsymbol{x}^{(0)}(\boldsymbol{x}^{(0)})^{H} \rVert_{F} \leq \rho \lVert \boldsymbol{x}\boldsymbol{x}^{H} \rVert_{F},
			\label{eq:dis}		
		\end{align}
		for some constant $\rho \in (0,1)$, provided that $\overline{\overline{\mathcal{I}^{c}_{0}}} \geq \mu N$, with $\mu>0$ sufficiently large.
		\label{theo:initi}
	\end{theorem}
	\begin{IEEEproof}
		See Appendix \ref{app:prooftheoini}.
	\end{IEEEproof}
	\begin{remark}
		It follows from Theorem \ref{theo:initi} and Algorithm \ref{alg:initialization} that our extended initialization procedure is more general than the one proposed in \cite{pinilla2024phase} for the conventional AF because Algorithm \ref{alg:initialization} converges from any starting point $\tilde{\boldsymbol{x}}^{(0)}$. This is not the case with the initialization in \cite{pinilla2024phase}. 
	\end{remark}
	
	Algorithm \ref{alg:algorithm} below summarizes the numerical optimization steps to estimate the signal $\boldsymbol{x}$ from its AF $\boldsymbol{A}$. To solve \eqref{eq:wavemax}, we employed the software library PhasePack \cite{chandra2019phasepack}, which contains efficient implementations of PR methods\footnote{A MATLAB implementation of this library is available at \url{https://github.com/tomgoldstein/phasepack-matlab}.}. The input to the algorithm is the AF $\boldsymbol{A}$ as in \eqref{eq:Ambiguity} (Line 1). Then, we set the constants $\mu$, $\tau>0$, and number of iterations $T$ (Line 2). The next step is to compute the approximation matrix $\boldsymbol{X}^{(0)}=\boldsymbol{x}^{(0)}(\boldsymbol{x}^{(0)})^{H}$ of the signal $\boldsymbol{x}$ by solving \eqref{eq:coorthoIniti} following a power iteration method strategy (Line 3). The main loop (Lines 6-13) is the \textit{f}ast \textit{a}daptive \textit{s}hrinkage/\textit{t}hresholding \textit{a}lgorithm (FASTA), an adaptive/accelerated forward-backward splitting method \cite{goldstein2014field}. The projection to the semidefinitive cone of matrices to satisfy positive semidefinite matrix constraint in \eqref{eq:wavemax} is performed subsequently (Lines 11-13). Finally, the method computes the estimation of $\boldsymbol{x}$ (Line 14).
	\begin{algorithm}[H]
		\caption{\textit{WaveMax: }Recovery of $\boldsymbol{x}$ from its AF $\boldsymbol{A}$}
		\label{alg:algorithm}
		\small
		\begin{algorithmic}[1]
			\State{\textbf{Input: }Data $\left\lbrace\boldsymbol{A}[\alpha,k]:k=0,\cdots,N-1, \alpha\in [-\pi/2,\pi/2] \right\rbrace$} \State{Choose the constants $\mu,\tau >0$, and number of iterations $T$}
			\Statex{}
			\State{$\boldsymbol{x}^{(0)}\leftarrow$ Algorithm \ref{alg:initialization}($\boldsymbol{A}[\alpha,k]$)}
			\State{Compute $\boldsymbol{X}^{(0)}=\boldsymbol{x}^{(0)}(\boldsymbol{x}^{(0)})^{H}$}
			\State{Set $\boldsymbol{Z}^{(0)} = \boldsymbol{X}^{(0)}$}
			\Statex{}
			\For{$r=0:T-1$} 
			\State{Compute $Ph_{\alpha,k}^{(r+1)} = \text{Tr}\left(\boldsymbol{B}_{\alpha,k}\boldsymbol{Z}^{(r)}\right)$}
			\State{Compute $\boldsymbol{G}^{(r+1)}[\alpha,k] = \frac{Ph^{(r+1)}}{\left\lvert  Ph^{(r+1)} \right\rvert}\left( \left\rvert Ph^{(r+1)} \right\rvert - \sqrt{\boldsymbol{A}[\alpha,k]}\right)$}
			\State{Compute $\boldsymbol{W}^{(r+1)} = \boldsymbol{Z}^{(r)} - \tau\boldsymbol{G}^{(r+1)}$}
			\State{Compute $\boldsymbol{S}^{(r+1)} = \boldsymbol{I}_{N} + \tau\mu \boldsymbol{W}^{(r+1)}$}
			\State{Eigenvalue decomposition  $\boldsymbol{S}^{(r+1)} = \boldsymbol{V}^{(r+1)}\boldsymbol{D}^{(r+1)}(\boldsymbol{V}^{(r+1)})^{H}$}
			\State{Compute $\hat{\boldsymbol{D}}^{(r+1)}[n,n] = \max(\boldsymbol{D}^{(r+1)}[n,n]-\mu\tau,0)$}
			\State{Set  $\boldsymbol{Z}^{(r+1)} = \boldsymbol{V}^{(r+1)}\hat{\boldsymbol{D}}^{(r+1)}(\boldsymbol{V}^{(r+1)})^{H}$}
			\Statex{}
			\EndFor
			
			\State{Compute $\boldsymbol{x}^{(T)}$ as leading eigenvector of $\boldsymbol{Z}^{(T)}$}
			\State{\textbf{return: } $\boldsymbol{x}^{(T)}$ \Comment{Estimation of $\boldsymbol{x}$}}
		\end{algorithmic}
	\end{algorithm}
	
	\subsection{WaveMax Convergence and Complexity Analysis}
	The convergence of Algorithm~\ref{alg:algorithm} -- which adopts the FASTA framework implemented in~\cite{chandra2019phasepack} -- is analytically guaranteed under convexity of the loss function and constraints in~\eqref{eq:wavemax}, as established by~\cite{goldstein2014field}. Lemma~\ref{lem:convexWavemax} confirms this convexity requirement, ensuring global convergence to a minimizer. Regarding computational efficiency, the dominant cost arises from the per-iteration Eigenvalue Decomposition, exhibiting $\mathcal{O}(n^3)$ complexity~\cite{pan1999complexity}. Previous works on FrFT-based AF PR lack computationally tractable algorithms for waveform recovery~\cite{pei2001relations,almeida1994fractional,jaming2010phase} or uniqueness guarantees~\cite{pinilla2024phase}. Algorithm~\ref{alg:algorithm} guarantees a unique solution up to global ambiguities even within a general phase retrieval framework.
	
	
	\subsection{Uniqueness analysis of WaveMax}
	To establish theoretical guarantees for the unique recovery of $\boldsymbol{x}\boldsymbol{x}^{H} \in \mathbb{C}^{N\times N}$ through the \textit{WaveMax} formulation in \eqref{eq:wavemax}, we recast the problem, solely for analytical purposes, as follows:
	\begin{align}
		\minimize_{\boldsymbol{Z}\in \mathbb{C}^{N\times N}} \hspace{0.5em} & -\text{Tr}\left(\boldsymbol{Z}^{H}\boldsymbol{X}^{(0)} \right) + \sum_{\alpha} i_{\mathbb{R}^{N}_{+}}(\boldsymbol{b}_{\alpha}) \nonumber\\
		\text{ subject to }\hspace{0.5em} & \left\lvert \sum_{n=0}^{N-1} \boldsymbol{b}_{\alpha}[n]w^{-k n} \right\rvert \leq \sqrt{\boldsymbol{A}[\alpha,k]} \nonumber\\
		&\text{Tr}(\boldsymbol{U}_{n,\alpha}
		\boldsymbol{Z}) = \boldsymbol{b}_{\alpha}[n] \nonumber\\
		& \boldsymbol{Z} \succeq 0,
		\label{eq:wavemaxAnalysis}
	\end{align}
	where $i_{\mathbb{R}^{N}_{+}}(\cdot)$ is the indicator function defined as
	\begin{align}
		i_{\mathcal{E}}(\boldsymbol{s}) = \left\{ \begin{array}{ll}
			0, & \text{ if } \boldsymbol{s}\in \mathcal{E}, \\
			\infty, & \text{otherwise},
		\end{array}\right.
	\end{align}
	and matrices $\boldsymbol{U}_{n,\alpha}$ are as in \eqref{eq:wavemax}. Then, it follows from \eqref{eq:wavemaxAnalysis} that it is enough to prove unique recovery by solving \eqref{eq:wavemax} if  $\boldsymbol{x}\boldsymbol{x}^{H}$ is the unique feasible point of $\left\{\boldsymbol{Z}\in \mathbb{C}^{N\times N}| \boldsymbol{Z} \succeq 0, \text{Tr}(\boldsymbol{U}_{n,\alpha} \boldsymbol{Z}) = \boldsymbol{b}_{\alpha}[n], \forall n,\alpha\right\}$. To guarantee this uniqueness result, we need to analyze the following linear mapping $\mathcal{B}:\mathcal{S}^{N\times N}\rightarrow \mathbb{R}_{+}^{N^{2}}$, where $\mathcal{S}^{n\times n}$ is the space of self-adjoint matrices, such that
	\begin{align}
		\mathcal{B}(\boldsymbol{S}) = \left[ \text{Tr}(\boldsymbol{U}_{0,\alpha_{0}} \boldsymbol{S}),\dots, \text{Tr}(\boldsymbol{U}_{N-1,\alpha_{N-1}} \boldsymbol{S})\right]^{T},
		\label{eq:auxOperator}
	\end{align}
	where the subscript for $\alpha$ indexes the number of angles. The following Lemma \ref{lem:solution} states the key property of the linear operator $\mathcal{B}$ to prove uniqueness of \textit{WaveMax} formulation in \eqref{eq:wavemax}. 
	\begin{lemma}
		Fix $\epsilon>0$. Then, under the setup of Theorem \ref{theo:initi}, we have
		\begin{align}
			\left\lVert \frac{1}{N}\mathcal{B}^{*}(\boldsymbol{1}) - \boldsymbol{I} \right\rVert_{2} \leq \epsilon,
		\end{align}
		where $\boldsymbol{1}$ is the all-one vector, provided that $\overline{\overline{\mathcal{I}^{c}_{0}}} \geq \mu N$, with $\mu>0$ sufficiently large.
		\label{lem:solution}
	\end{lemma}
	\begin{IEEEproof}
		From Theorem \ref{theo:initi},
		\begin{align}
			\frac{1}{N}\mathcal{B}^{*}(\boldsymbol{1}) &= \frac{1}{N} \sum_{\alpha,n=0}^{N-1} \boldsymbol{u}_{n,\alpha}\boldsymbol{u}_{n,\alpha}^{H}  = \boldsymbol{I},
		\end{align}
		holds provided that $\overline{\overline{\mathcal{I}^{c}_{0}}} \geq \mu N$, with $\mu>0$ sufficiently large. For simplicity but without loss of generality, assume each $\boldsymbol{u}_{n,\alpha}$ is unitary. Then,
		\begin{align}
			\left\lVert \frac{1}{N}\mathcal{B}^{*}(\boldsymbol{1}) - \boldsymbol{I} \right\rVert_{2} \leq \epsilon,
		\end{align}
		for any $\epsilon\in (0,1)$.
	\end{IEEEproof}
	
	Taking into account the linear operator $\mathcal{B}(\cdot)$ in \eqref{eq:auxOperator}, a band-limited signal $\boldsymbol{x}$ can be recovered from its AF $\boldsymbol{A}$ if $\mathcal{B}(\cdot)$ is injective \cite{candes2013phase}. With an accurate estimation $\boldsymbol{x}^{(0)}$ of $\boldsymbol{x}$, we adopt a strategy similar to that in \cite{candes2013phase,gross2017improved}. The following Theorem~\ref{theo:guarantee} outlines the conditions that the operator $\mathcal{B}$ must satisfy for guaranteed recovery through solving \eqref{eq:wavemax} \cite{candes2013phase}. However, unlike the approach in \cite{candes2013phase}, we must ensure that $\mathcal{B}$ can uniquely identify $\boldsymbol{x}\boldsymbol{x}^{H}$, as the operator's effectiveness relies on the close estimation $\boldsymbol{x}^{(0)}$ of $\boldsymbol{x}$. Further, the non-probabilistic isometries discussed in Theorem~\ref{theo:guarantee} have not been addressed in previous literature on Fourier-based PR problems, necessitating a new proof.
	\begin{theorem}
		Consider the operator $\mathcal{B}:\mathcal{S}^{N\times N}\rightarrow \mathbb{R}_{+}^{N^{2}}$ as defined in \eqref{eq:auxOperator}. Define the tangent space of the manifold of all rank-1 Hermitian matrices at the point $\boldsymbol{x}\boldsymbol{x}^{H}$ as
		\begin{align}
			\mathcal{T}_{\boldsymbol{x}} = \left\lbrace \boldsymbol{x}\boldsymbol{s}^{H} + \boldsymbol{s}\boldsymbol{x}^{H}| \boldsymbol{s}\in \mathbb{C}^{N} \right \rbrace.
			\label{eq:tangent}
		\end{align} 
		Under the setup of Theorem \ref{theo:initi} for almost all $B$-band-limited signals $\boldsymbol{x}\in \mathbb{C}^{N}$, with $B\leq N/2$, we have
		\begin{description}
			\item[A1] The operator $\mathcal{B}$ is injective and satisfies
			\begin{equation}
				(1-\delta)\lVert \boldsymbol{S}\rVert_{1}\leq\frac{1}{N}\lVert\mathcal{B}(\boldsymbol{S})\rVert_{1}\leq(1+\delta)\lVert\boldsymbol{S}\rVert_{1},
				\label{eq:comple}
			\end{equation}
			for all matrices $\boldsymbol{S}\in \mathcal{T}_{\boldsymbol{x}}$ provided that $\overline{\overline{\mathcal{I}^{c}_{0}}} \geq \mu N$, with $\mu>0$ sufficiently large, for some constant $\delta\in(0,1)$.
			\item[A2] There exists a self-adjoint matrix of the form $\boldsymbol{Q}=\mathcal{B}^{*}(\boldsymbol{\lambda})$, where self-adjointness implies $\boldsymbol{\lambda} \in \mathbb{R}^{N^{2}}$, such that
			\begin{align}
				\boldsymbol{Q}_{\mathcal{T}^{\perp}_{\boldsymbol{x}}} \preceq - \boldsymbol{I}_{\mathcal{T}^{\perp}_{\boldsymbol{x}}}, \text{ and } \lVert \boldsymbol{Q}_{\mathcal{T}_{\boldsymbol{x}}}\rVert_{F} \leq \zeta,
				\label{eq:certificate}
			\end{align}
			where $\zeta\in (0,1)$, the matrix $\boldsymbol{Q}_{\mathcal{T}^{\perp}_{\boldsymbol{x}}}$ denotes the orthogonal projection of $\boldsymbol{Q}$ onto $\mathcal{T}^{\perp}_{\boldsymbol{x}}$ given by $\boldsymbol{Q}_{\mathcal{T}^{\perp}_{\boldsymbol{x}}} = (\boldsymbol{I}-\boldsymbol{x}\boldsymbol{x}^{H})\boldsymbol{Q}(\boldsymbol{I}-\boldsymbol{x}\boldsymbol{x}^{H})$, and $\mathcal{B}^{*}(\cdot)$ represents the adjoint operator of $\mathcal{B}(\cdot)$.
		\end{description}
		Then, it follows from \textbf{A1} and \textbf{A2} that $\boldsymbol{x}\boldsymbol{x}^{H}$ is the unique element that satisfies the inequality constrains in \eqref{eq:wavemaxAnalysis}. Therefore, $\boldsymbol{x}\boldsymbol{x}^{H}$ also uniquely satisfies the inequality constraints in \eqref{eq:wavemax}.
		\label{theo:guarantee}
	\end{theorem}
	\begin{IEEEproof}
		See Appendix \ref{app:guarantees}.
	\end{IEEEproof}
	
	Theorem \ref{theo:guarantee} establishes that a band-limited signal $\boldsymbol{x}$ is uniquely determined from the linear operator $\mathcal{B}$. Condition \textbf{A1}, 
	akin to the local restricted isometry property, ensures the \textit{robust injectivity} of the mapping $\mathcal{B}$ restricted to elements in $\mathcal{T}_{\boldsymbol{x}}$. 
	Meanwhile, \textbf{A2} indicates the existence of an \textit{approximate dual certificate}. In general, injectivity paired with an exact dual certificate enables exact reconstruction. Thus, Theorem \ref{theo:guarantee} demonstrates that a robust form of injectivity, combined with an approximate dual certificate, guarantees exact recovery, as discussed in \cite[Section 2.1]{candes2013phaselift} and further explored in \cite{gross2011recovering}. 
	
	\section{Numerical Experiments}
	\label{sec:results}
	\begin{figure*}[t]
		\centering
		\includegraphics[width=1\linewidth]{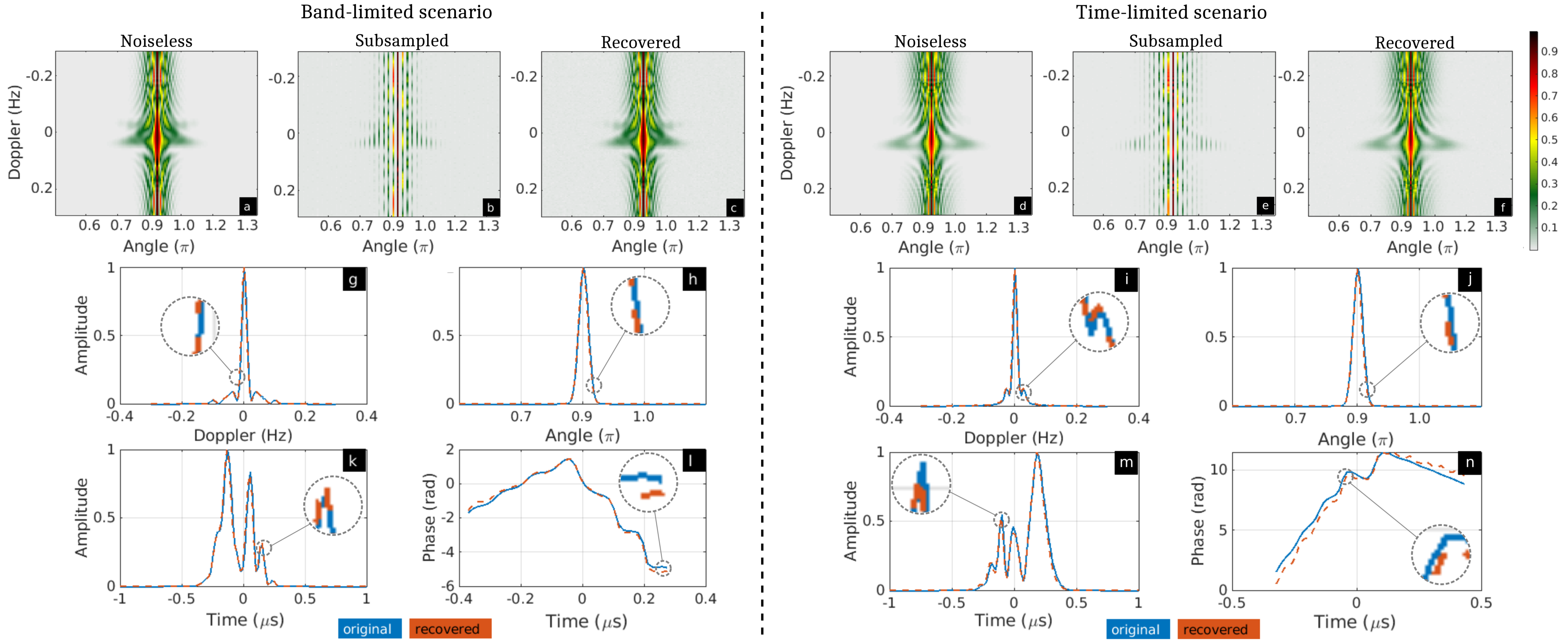}
		\vspace{-1.5em}
		\caption{Reconstructed band- and time-limited signals when 75\% of the samples of their AFs are uniformly removed. The sparse AFs were corrupted by noise such that SNR = $20$ dB. The attained relative error as in \eqref{eq:distance} was $5\times 10^{-2}$ for both signals. For the band-limited [time-limited] signal, (a) [(d)],(b) [(e)], and (c) [(f)] are the original, sub-sampled, and recovered AFs, respectively; (g) [(i)], and (h) [(j)] are 1-D slices of AFs in the angle and Doppler dimensions, respectively; (k) [(m)] and (l) [(n)] are, respectively, magnitude and phase of recovered (red) signal juxtaposed over the original (blue).
		}
		\label{fig:incomplete_noisyresults1}
		\vspace{-1em}
	\end{figure*}
	
	\begin{figure*}[t]
		\centering
		\includegraphics[width=1\linewidth]{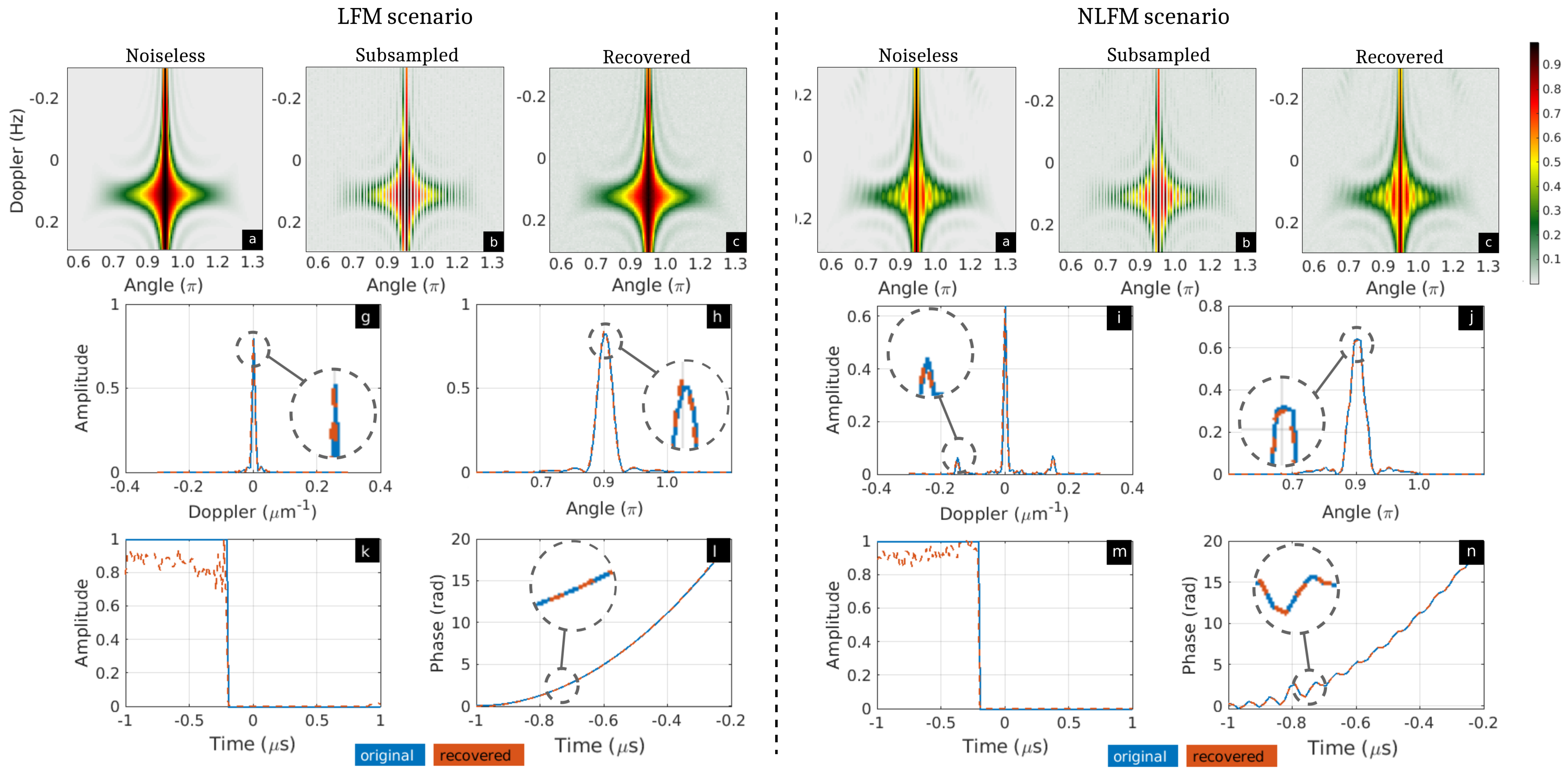}
		\vspace{-2em}
		\caption{Reconstructed LFM and NLFM signals when 75\% of the samples of their AFs are uniformly removed. The sparse AFs were corrupted by noise such that SNR = $20$ dB. The attained relative error as in \eqref{eq:distance} was $5\times 10^{-2}$ for both signals. For the LFM [NLFM] signal, (a) [(d)],(b) [(e)], and (c) [(f)] are the original, sub-sampled, and recovered AFs, respectively; (g) [(i)], and (h) [(j)] are 1-D slices of AFs in the angle and Doppler dimensions, respectively; (k) [(m)] and (l) [(n)] are, respectively, magnitude and phase of recovered (red) signal juxtaposed over the original (blue).
		}
		\label{fig:incomplete_noisyresults2}
		\vspace{-1em}
	\end{figure*}
	
	We now study the performance of Algorithm \ref{alg:algorithm} to estimate the signal of interest $\boldsymbol{x}$ from its FrFT-based AF in terms of relative error as \eqref{eq:distance}. We also compare Algorithm \ref{alg:algorithm}, which is a convex approach to solve the radar PR problem, with the previous non-convex strategy in \cite{pinilla2024phase} that follows a short-time Fourier transform (STFT)-based PR formulation. 
	We built a set of $\left\lceil \frac{N-1}{2} \right\rceil$-band-limited signals that conform to a Gaussian power spectrum. 
	Specifically, each signal ($N=128$ grid points) is produced via the FT of a complex vector with a Gaussian-shaped amplitude. 
	Next, we multiply the obtained power spectrum by a uniformly distributed random phase\footnote{All simulations were implemented in Matlab R2021a on an Intel Core i5 2.5 GHz CPU with 16 GB RAM.}. We normalize the Doppler and delays to $[0,1]$.

	We evaluate the performance of the proposed method under noisy and noiseless scenarios for complete or sparse samples of the AF at different values of signal-to-noise-ratio (SNR), defined as SNR$ = 10\log_{10}(\lVert \boldsymbol{A} \rVert^{2}_{\mathcal{F}}/\lVert \boldsymbol{\sigma} \rVert^{2}_{\text{2}})$, where $\boldsymbol{\sigma}$ is the variance of the noise. The sparse samples of the AF were uniformly chosen. 
	\vspace{0.5em}
	
	\noindent\textbf{Recovery examples}: To demonstrate the effectiveness of solving \eqref{eq:wavemax}, Figure \ref{fig:incomplete_noisyresults1} presents the estimated band- and time-limited signals from a noisy sparse AF. Here, $75$\% of the samples from the AF are uniformly removed, retaining every third angle (dimension $\alpha$) starting from the first. The AF was corrupted by white noise with SNR = $20$ dB. These results provide numerical validation for Proposition~\ref{prop:uniqueness}, indicating that not all AF samples are necessary to estimate the underlying signal. The solution to \eqref{eq:wavemax} closely approximates $\boldsymbol{x}$, even under the assumption of an imperfectly designed sparse AF, highlighting the effectiveness of \eqref{eq:wavemax} for the FrFT-based AF PR problem.\vspace{0.5em}
	
	\noindent\textbf{Modulated signals}: We now investigate the performance of Algorithm \ref{alg:algorithm} in estimating structured signals from their sparse noisy AFs. Here, we consider the LFM and NLFM waveforms,
	\begin{align}
		\boldsymbol{x}[n] = \boldsymbol{a}[n]e^{i\pi \varphi[n]},
	\end{align}
	where 
	\begin{align}
		\varphi[n] &= \begin{dcases} \pi r (\Delta tn)^{2}, \hspace{0.5em}&(\text{LFM}), \nonumber\\
			\pi r t^{2} + \sum_{l=1}^{L}\beta_{l}\cos(2\pi l \Delta tn/T), \hspace{0.5em}&(\text{NLFM}),\end{dcases}
	\end{align}
	with $T$ as the duration of the pulse, $\Delta t$ as the sampling size in time, $r=\frac{\Delta f}{T}$ such that $\Delta f$ is the swept bandwidth, and $L>0$ is an integer, and $\beta_{l} = \frac{0.4T}{l}$. We set $\Delta f = 128\times 10^{3}$, $\Delta t = 0.4\times 10^{-6}$, and 
	\begin{equation}
		\boldsymbol{a}[n] = \left \lbrace \begin{array}{ll}
			1, & 0\leq \Delta tn \leq T, \\
			0, & \text{otherwise}.
		\end{array} \right.
	\end{equation}
	We uniformly removed 50\% of the samples from the AF. Figure \ref{fig:incomplete_noisyresults2} shows the recovered signal for SNR $=20$ dB. 
	We observe that Algorithm \ref{alg:algorithm} is able to estimate the phase of the pulses accurately, while the reconstructed magnitudes show some artifacts; note that we limited the reconstruction quality here by taking only 50\% AF measurements. Again, Proposition~\ref{prop:uniqueness} is numerically validated, and not all AF samples are needed to estimate $\boldsymbol{x}$.\vspace{0.5em}
	
	Note that the reconstructed waveforms in Figures~\ref{fig:incomplete_noisyresults1}, and \ref{fig:incomplete_noisyresults2} are obtained from severely undersampled AF measurements, where 75\% of the samples were uniformly excised, compounded by additive noise contamination. Under such conditions, perfect waveform reconstruction is extremely challenging, leading to observable artifacts in either phase coherence or amplitude fidelity depending on the waveform's structure. For the examined LFM and NLFM waveforms, the dominant manifestation is amplitude distortion -- a consequence of their inherent sensitivity to sparse AF sampling and noise perturbation during the reconstruction process. However, for the waveforms in Figure~\ref{fig:incomplete_noisyresults1}, the main distortion is observed in the phase. This phenomenon arises because missing AF samples disproportionately disrupt the waveform estimation.
	
	\begin{figure}[t]
		\centering
		\includegraphics[width=1\linewidth]{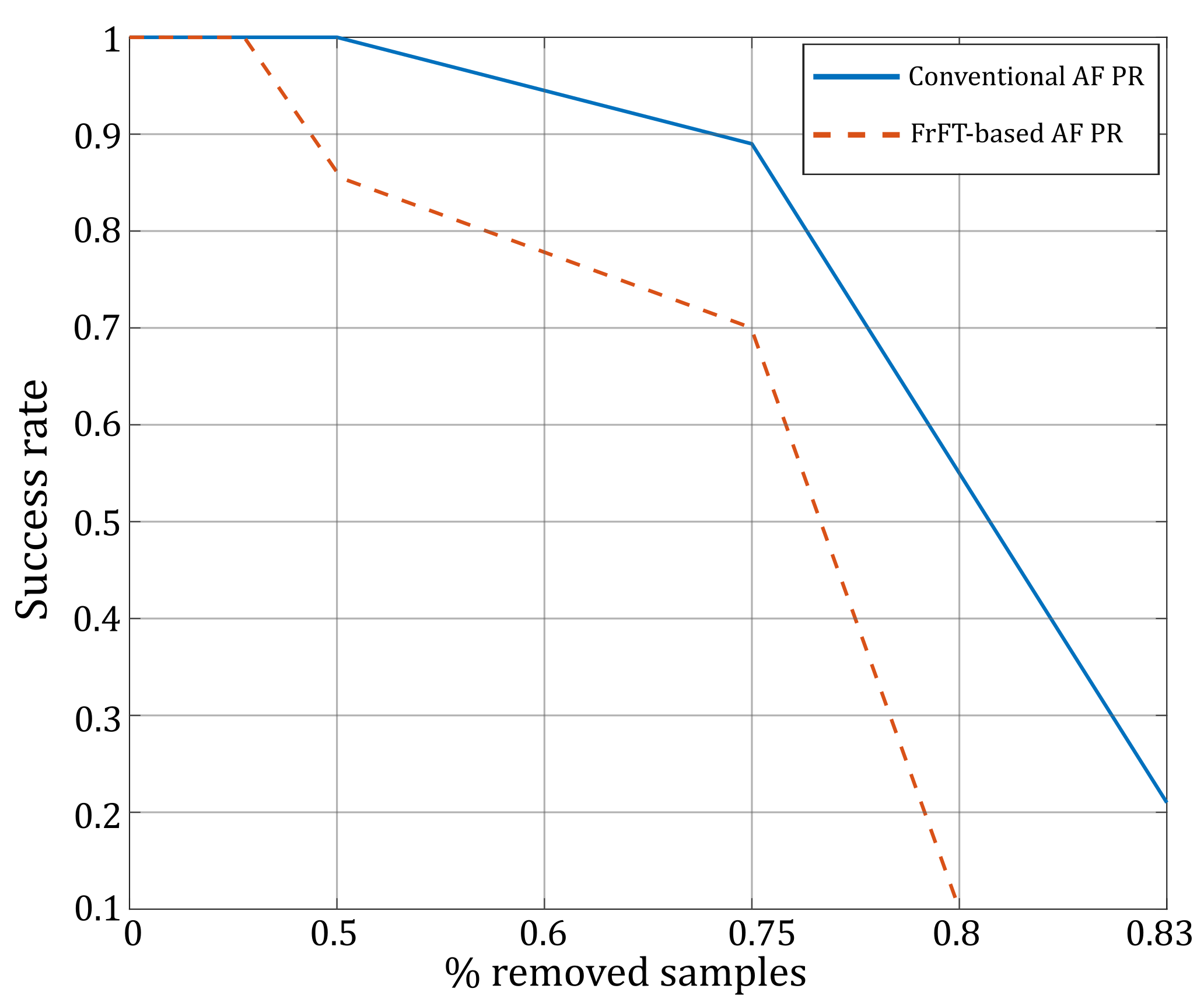}
		\caption{Empirical success rate of solving \eqref{eq:wavemax} for complete and sparse AF samples in the absence of noise using the proposed Algorithm \ref{alg:algorithm} for FrFT-based AF and the non-convex BanRaW method of \cite{pinilla2024phase} for the conventional AF.}
		\label{fig:success}
		\vspace{-1em}
	\end{figure}
	
	\noindent\textbf{Statistical performance}:
	We study the empirical success rate of solving \eqref{eq:wavemax} through Algorithm \ref{alg:algorithm} for complete and sparse samples of the AF in the absence of noise. 
	We declare a trial successful when the returned estimate attains a relative error in \eqref{eq:distance} smaller than $1\times 10^{-6}$. The success rate and the number of iterations are averaged over 100 pulses, where $20\%$ of these pulses are LFM and NLFM ($10\%$ each class) waveforms as described in the case of modulated signals above. Figure \ref{fig:success} compares the success rates of Algorithm \ref{alg:algorithm} and non-convex BanRaW \cite{pinilla2024phase} to fully reconstruct the signal $\boldsymbol{x}$ for complete ($0$\% removed samples) and sparse AF measurements. Although the probability of success of the convex Algorithm \ref{alg:algorithm} is smaller than the nonconvex method \cite{pinilla2024phase}, the upshot of Algorithm \ref{alg:algorithm} is that it theoretically guarantees a unique solution as in Theorem \ref{theo:guarantee}. This guarantee is absent in~\cite{pinilla2024phase}.\vspace{0.5em}
	
	\noindent\textbf{Initialization procedure error}: 
	Finally, we compare the performance of the vector $\boldsymbol{x}^{(0)}$ obtained by Algorithm \ref{alg:initialization} against the initialization procedure introduced in \cite{pinilla2024phase} to approximate the pulse $\boldsymbol{x}$. We analyze this under both noisy and noiseless scenarios for complete as well as sparse samples of the AF. Since \eqref{eq:finalInit} involves the computation of the leading eigenvector of a matrix, recall that Algorithm \ref{alg:initialization} follows a power iteration strategy to estimate $\boldsymbol{x}^{(0)}$. We fixed the number of iterations of this method at $100$. We numerically determine the relative error in \eqref{eq:distance} by averaging it over 100 trials. Figure \ref{fig:init} shows that the initialization procedure proposed in Algorithm \ref{alg:initialization} produces a more accurate approximation of $\boldsymbol{x}$ than BanRaW using the same number of samples because the returned relative error of the former is smaller than that obtained using the latter. This validates the analytical result that \eqref{eq:finalInit} does not require all samples of $\boldsymbol{A}$ to approximate $\boldsymbol{x}$, a fact that is not theoretically guaranteed in \cite{pinilla2024phase}. Further, it is worth mentioning that in Algorithm \ref{alg:initialization} the leading eigenvector is derived from the matrix
		\begin{align*}
			\boldsymbol{G}_{0} = \frac{1}{ \overline{\overline{\mathcal{I}^{c}_{0}}}} \sum_{(n,\alpha) \in \mathcal{I}^{c}_{0}} \boldsymbol{u}_{n,\alpha}\boldsymbol{u}_{n,\alpha}^{H}
		\end{align*}
		which is structurally decoupled from the noisy AF. As a result, while noise perturbs the estimation of the index set $\mathcal{I}^{c}_{0}$, it does not have a direct influence on the spectral decomposition of $\boldsymbol{G}_{0}$. Consequently, the eigenvector computation remains invariant to noise corruption, preserving the fidelity of the waveform recovery. Isolating the eigendecomposition from the estimation of the noise-dependent set implies that the Algorithm \ref{alg:initialization} is more robust to noise than the strategy proposed in~\cite{pinilla2024phase}.
	
	\begin{figure}[ht]
		\centering
		\includegraphics[width=1\linewidth]{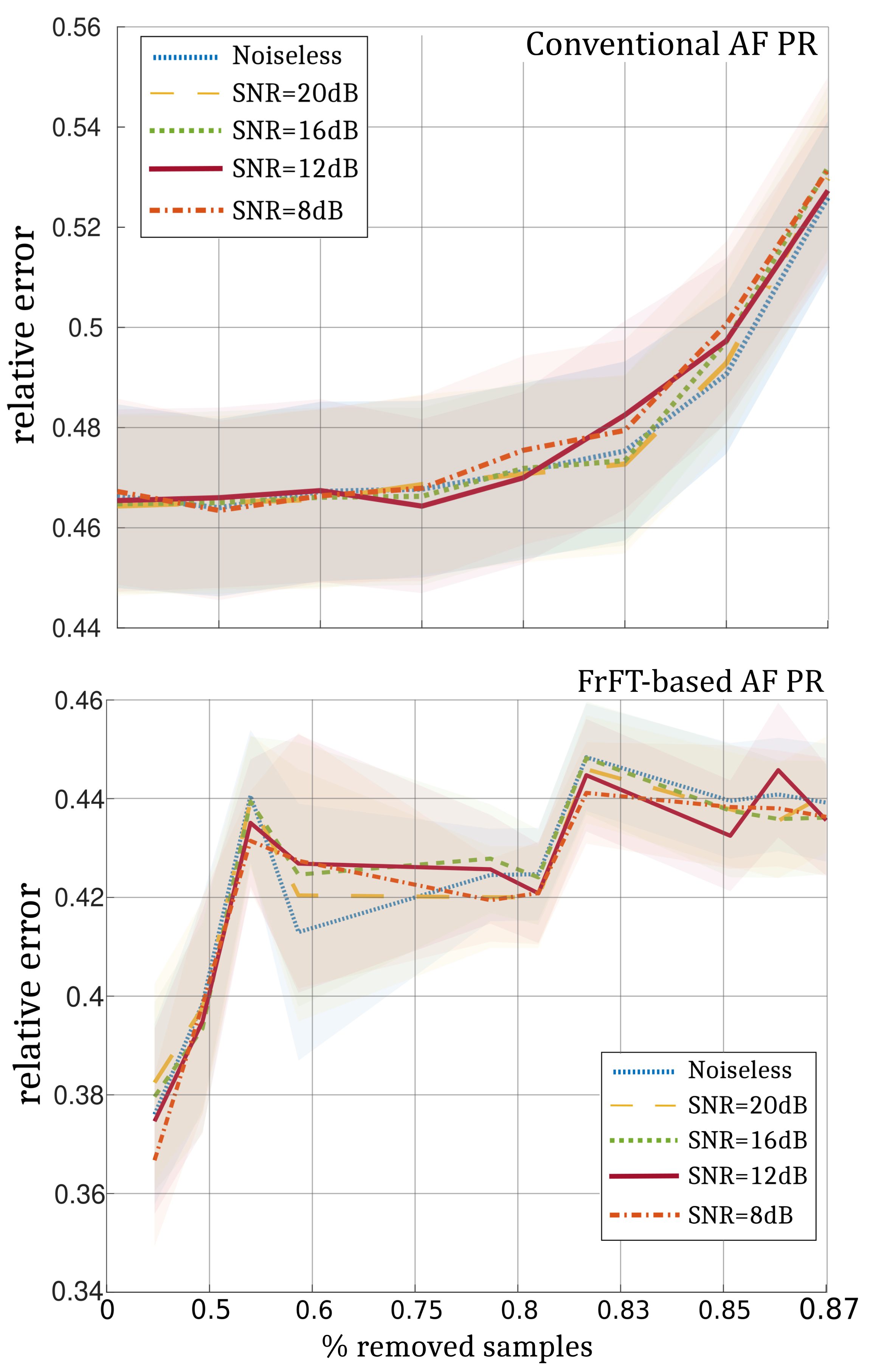}
		\caption{Relative error using the initialization procedures of Algorithm \ref{alg:initialization} and BanRaW \cite{pinilla2024phase} to estimate the underlying signal $\boldsymbol{x}$ at different SNRs plotted with respect to different percentages of uniformly removed delays in \cite{pinilla2024phase} and rotations in Algorithm \ref{alg:initialization}. The shaded background regions represents the variance of the relative error over 100 trials, while the solid lines are the mean.}
		\label{fig:init}
		\vspace{-1em}		
	\end{figure}
	\section{Summary}
	\label{sec:conclusion}
	While both AF and FrFT belong to specific classes of quadratic TFRs, the FrFT-based AF allows further design possibilities. Other generalization of FrFT such as special affine FTs \cite{wolf1979integral,abe1994optical,bernardo1996abcd} offer even more degrees of freedom, but their PR guarantees are difficult to come by. Other TFR kernels (see, e.g., \cite[Tables 3.3.2, 3.3.3, 6.1.2]{boashash2016time} and \cite{hlawatsch1992linear} for a detailed list) such as WVD, Rihaczek \cite{rihaczek1968signal}, Levin \cite{levin2006instantaneous}, Choi-Williams \cite{choi1989improved}, Zhao-Atlas-Marks \cite{zhao1990use}, Born-Jordan \cite{born1925zur}, and Cheriet-Belouchrani kernel \cite{abed2012time} provide valuable insights. However, they are hindered by issues such as cross-term interference, a trade-off between time and frequency resolution, less accurate signal representations, and challenges in developing efficient PR recovery algorithms.
	
	In this regard, we analytically demonstrated that our approach recovers band/time-limited signals, up to trivial ambiguities, from their FrFT-based AFs. Our convex optimization problem to estimate these signals under complete/sparse noisy and noiseless scenarios is an improvement over previous approaches. It requires a designed approximation of the radar signal obtained by extracting the leading eigenvector of a matrix depending on the AF. This approximation is employed to estimate the underlying signal satisfying a convex relaxation of the measurement constraints. Further, in the case of sparse data, our Proposition~\ref{prop:uniqueness} suggests that the full AF is not required to guarantee uniqueness. Numerical experiments showed robust recovery of both the magnitude and phase of the signal even from noisy sparse data. This technique paves way for addressing AF-based radar PR for more complex TFR kernels.

	\appendices
	
	\section{Proof that $\mathcal{T}(\boldsymbol{x})$ is closed}
	\label{app:distance}
	If $\mathcal{T}(\boldsymbol{x})$ is closed, then it guarantees the existence of $\boldsymbol{z}\in \mathcal{T}(\boldsymbol{x})$ such that 
	\begin{align}
		\text{dist}(\boldsymbol{x},\boldsymbol{s}) = \lVert \boldsymbol{s} - \boldsymbol{z} \rVert_{2},
		\label{eq:wellDefined}
	\end{align}
	for any $\boldsymbol{s}\in \mathbb{C}^{N}$. We have the following Lemma~\ref{lem:closed} that guarantees the existence of the minimum of $\lVert \boldsymbol{q}-\boldsymbol{z}\rVert_{2}$ in \eqref{eq:distance}.
	\begin{lemma}\label{lem:closed}
		The set 
		\begin{align}
			\mathcal{T}(\boldsymbol{x})=&\left\lbrace\boldsymbol{z}\in \mathbb{C}^{N} \hspace{0.2em}:\hspace{0.2em} \boldsymbol{z}[n]=e^{i\beta}e^{i b n}\boldsymbol{x}[\epsilon n - a] \text{ for } \beta,b\in \mathbb{R}, \right.\nonumber\\
			&\left. \text{ and }\epsilon=\pm 1, a\in \mathbb{Z} \right \rbrace,
			\label{eq:levelSet}
		\end{align}
		is closed for $\boldsymbol{x}\in \mathbb{C}^{N}$.
	\end{lemma}
	
	\begin{proof}
		Observe that a given element $\boldsymbol{z}$ in the set $\mathcal{T}(\boldsymbol{x})$ is the product between an orthogonal matrix and the signal $\boldsymbol{x}$. This is concluded as follows.
		\begin{itemize}
			\item The time inversion on $\boldsymbol{x}$ is achieved by the product of an orthogonal matrix $\boldsymbol{S}_{1}\in \mathbb{R}^{N\times N}$ (which depends on $\epsilon\pm 1$), and $\boldsymbol{x}$. When $\epsilon=1$, $\boldsymbol{S}_{1}$ is an identity matrix of size $N$; otherwise, it is a permutation matrix that models time reversal, i.e. $\boldsymbol{x}[-n]$.
			
			\item A constant time shift of $a$ on the signal $\boldsymbol{x}$ according to $\mathcal{T}(\boldsymbol{x})$ is modelled by the product of a circulant and orthogonal) matrix $\boldsymbol{S}_{2}\in \mathbb{R}^{N\times N}$, which depends on the translation parameter $a$, and $\boldsymbol{x}$. 
			
			\item The time scale ambiguity in $\mathcal{T}(\boldsymbol{x})$ is modelled by the product of a diagonal (and, therefore, orthogonal) matrix $\boldsymbol{S}_{3}\in \mathbb{C}^{N\times N}$, which depends on the parameter $b$, and $\boldsymbol{x}$. The matrix $\boldsymbol{S}_{3}$ is orthogonal because it takes values from the complex unit circle.
			
			\item The global phase shift is a constant taken from the complex unit circle governed by $\beta$.
		\end{itemize}
		It follows from above that a given element 
		\begin{align}
			\boldsymbol{z} = \boldsymbol{Q}\boldsymbol{x}, 
			\label{eq:auxExpression}
		\end{align}
		where the orthogonal matrix
		\begin{align}
			\boldsymbol{Q} = e^{i\beta}\boldsymbol{S}_{s}\boldsymbol{S}_{r}\boldsymbol{S}_{u} \in \mathbb{C}^{N\times N},
			\label{eq:matrixQ}
		\end{align}
		for $s,r,u\in \{1,2,3\}$. Note that the values of $s,r,u$, which determine the product order, directly depend on $\epsilon$, $b$, and $a$.
		
		To prove that the set $\mathcal{T}(\boldsymbol{x})$ is closed,  
		take a sequence $(\boldsymbol{z}_{n})\subset \mathcal{T}(\boldsymbol{x})$ that converges to $\boldsymbol{z}_{*}$. We have to show that $\boldsymbol{z}_{*}\in \mathcal{T}(\boldsymbol{x})$. Since $\boldsymbol{z}_{n}\rightarrow \boldsymbol{z}_{*}$, then for any $\delta>0$, there exists a natural number $n(\delta)$ such that
		\begin{align}
			\lVert \boldsymbol{z}_{n} - \boldsymbol{z}_{*} \rVert_{2} < \delta, \hspace{0.5em} \forall n>n(\delta).
			\label{eq:closeSet}
		\end{align}
		Each $\boldsymbol{z}_{n}$ belongs to $\mathcal{T}(\boldsymbol{x})$. Hence, from \eqref{eq:auxExpression}, there exists $\boldsymbol{Q}_{n}\in \mathbb{C}^{N\times N}$ orthogonal matrix with the form of \eqref{eq:matrixQ} such that $\boldsymbol{z}_{n}=\boldsymbol{Q}_{n}\boldsymbol{x}$. Consequently, from \eqref{eq:closeSet}, for any $\delta>0$
		\begin{align}
			\lVert \boldsymbol{z}_{n} - \boldsymbol{z}_{*} \rVert_{2} = \lVert \boldsymbol{x} - \boldsymbol{Q}_{n}^{H}\boldsymbol{z}_{*} \rVert_{2} < \delta , \hspace{0.5em} \forall n>n(\delta),
		\end{align}
		for some natural number $n(\delta)$ because $\boldsymbol{Q}_{n}$ is orthogonal. The above equation \textit{ipso facto} means that the sequence $\boldsymbol{r}_{n}=\boldsymbol{Q}_{n}^{H}\boldsymbol{z}_{*}$ converges to $\boldsymbol{x}$. 	This leads to the existence of $\boldsymbol{Q}_{*}\in \mathbb{C}^{N\times N}$ following the form of \eqref{eq:matrixQ} such that $\boldsymbol{x}=\boldsymbol{Q}_{*}^{H}\boldsymbol{z}_{*}$. Equivalently, we obtain that $\boldsymbol{z}_{*} = \boldsymbol{Q}_{*}\boldsymbol{x}$ implying that $\boldsymbol{z}_{*}\in \mathcal{T}(\boldsymbol{x})$ because $\boldsymbol{Q}_{*}$ follows the form of \eqref{eq:matrixQ}. This completes the proof.
	\end{proof}

	\section{Proof of Lemma \ref{lem:convexWavemax}}
	\label{app:convexWavemax}
	In order to prove optimization problem in \eqref{eq:wavemax} is convex, we have to show the set $$\mathcal{O}=\left\{\boldsymbol{Z}\in \mathbb{C}^{N\times N} | \left\lvert  \text{Tr}\left(\boldsymbol{B}_{\alpha,k}\boldsymbol{Z}\right) \right\rvert \leq \sqrt{\boldsymbol{A}[\alpha,k]}, \forall n,\alpha \right\},$$ is convex, for matrices $\boldsymbol{B}_{\alpha,k}$ as defined in \eqref{eq:matricesB}. Then, take $\boldsymbol{Z}_{1},\boldsymbol{Z}_{2}\in \mathcal{O}$ and $\lambda \in (0,1)$. Define, $\boldsymbol{Z} = \lambda \boldsymbol{Z}_{1} + (1-\lambda)\boldsymbol{Z}_{2}$. Observe that 
	\begin{align}
		\left\lvert  \text{Tr}\left(\boldsymbol{B}_{\alpha,k}\boldsymbol{Z}\right) \right\rvert &= \left\lvert  \text{Tr}\left(\boldsymbol{B}_{\alpha,k}(\lambda \boldsymbol{Z}_{1} + (1-\lambda)\boldsymbol{Z}_{2})\right) \right\rvert \nonumber\\
		&= \left\lvert  \lambda\text{Tr}\left(\boldsymbol{B}_{\alpha,k}\boldsymbol{Z}_{1}\right) + (1-\lambda)\text{Tr}\left(\boldsymbol{B}_{\alpha,k}\boldsymbol{Z}_{2}\right) \right\rvert \nonumber\\
		&\leq \lambda \left\lvert  \text{Tr}\left(\boldsymbol{B}_{\alpha,k}\boldsymbol{Z}_{1}\right) \right\rvert + (1-\lambda) \left\lvert  \text{Tr}\left(\boldsymbol{B}_{\alpha,k}\boldsymbol{Z}_{2}\right) \right\rvert \nonumber\\
		&\leq \sqrt{\boldsymbol{A}[\alpha,k]},
	\end{align}
	where the last inequality follows because  $\boldsymbol{Z}_{1},\boldsymbol{Z}_{2}\in \mathcal{O}$. It follows from the above equation that $\boldsymbol{Z} \in \mathcal{O}$ leading to the convexity of $\mathcal{O}$. QEF.
	
	\section{Proof of Theorem \ref{theo:initi}}
	\label{app:prooftheoini}
	\subsection{Preliminaries to the proof}
	Due to homogeneity in \eqref{eq:dis}, it suffices to work with the case where $\lVert \boldsymbol{x} \rVert_{2}=1$. Instrumental in proving Theorem \ref{theo:initi} is the following Lemma~\ref{lem:l2}.
	\begin{lemma}
		\label{lem:l2}
		Consider the noiseless data $\boldsymbol{Y}[\alpha,\ell]$. For almost all unit $B$-band-limited signals $\boldsymbol{x}\in \mathbb{C}^{N}$ with $B\leq N/2$, there exists a vector $\boldsymbol{u}\in \mathbb{C}^{N}$ with $\boldsymbol{u}^{H}\boldsymbol{x}=0$ and $\lVert \boldsymbol{u} \rVert_{2}=1$, such that
		\begin{align}
			\frac{1}{2}\lVert \boldsymbol{x}\boldsymbol{x}^{H}-\boldsymbol{x}^{(0)}(\boldsymbol{x}^{(0)})^{H} \rVert_{F}^{2} \leq \frac{\lVert \boldsymbol{S}\boldsymbol{u} \rVert_{2}^{2}}{\lVert \boldsymbol{S}\boldsymbol{x}\rVert_{2}^{2}},
			\label{eq:init}
		\end{align}
		where $\boldsymbol{S} = \frac{1}{N} \left[\boldsymbol{u}_{n_{1},\alpha_{1}},\cdots,\boldsymbol{u}_{n_{J},\alpha_{J}} \right]^{H}$ for $(n_{p},\alpha_{p}) \in \mathcal{I}_{0}^{c}$, $J=\overline{\overline{\mathcal{I}^{c}_{0}}}$, and $p=1,\dots,J$.
		\begin{IEEEproof}
			Note that
			\begin{align}
				\frac{1}{2}\lVert \boldsymbol{x}\boldsymbol{x}^{H}-\boldsymbol{x}^{(0)}(\boldsymbol{x}^{(0)})^{H} \rVert_{F}^{2} &= \frac{1}{2}\lVert \boldsymbol{x} \rVert_{2}^{4} + \frac{1}{2}\lVert \boldsymbol{x}^{(0)} \rVert_{2}^{4} - \lvert \boldsymbol{x}^{H}\boldsymbol{x}^{(0)} \rvert^{2} \nonumber
				\\
				&=1- \lvert \boldsymbol{x}^{H}\boldsymbol{x}^{(0)} \rvert^{2} = 1-\cos^{2}(\theta),
				\label{eq:cos}
			\end{align}
			where $\theta\in [0,\pi/2]$ is the angle between the spaces spanned by $\boldsymbol{x}$ and $\boldsymbol{x}^{(0)}$. Then,
			\begin{align}
				\boldsymbol{x} = \cos(\theta)\boldsymbol{x}^{(0)}+\sin(\theta)(\boldsymbol{x}^{(0)})^{\perp},
				\label{eq:p1}
			\end{align}
			where $(\boldsymbol{x}^{(0)})^{\perp}\in \mathbb{C}^{N}$ is a unit vector orthogonal to $\boldsymbol{x}^{(0)}$ and the real part of its inner product with $\boldsymbol{x}$ is non-negative. From \eqref{eq:p1}, we have
			\begin{align}
				\boldsymbol{x}^{\perp} = -\sin(\theta)\boldsymbol{x}^{(0)}+\cos(\theta)(\boldsymbol{x}^{(0)})^{\perp},
				\label{eq:p2}
			\end{align}
			in which $\boldsymbol{x}^{\perp}\in \mathbb{C}^{N}$ is a unit vector orthogonal to $\boldsymbol{x}$. Thus, considering \eqref{eq:p1} and \eqref{eq:p2} and then appealing to \cite[Lemma 1]{wang2018solving} yields
			\begin{align}
				\frac{1}{2}\lVert \boldsymbol{x}\boldsymbol{x}^{H}-\boldsymbol{x}^{(0)}(\boldsymbol{x}^{(0)})^{H} \rVert_{F}^{2} \leq \frac{\lVert \boldsymbol{S}\boldsymbol{x}^{\perp} \rVert_{2}^{2}}{\lVert \boldsymbol{S}\boldsymbol{x}\rVert_{2}^{2}}.
				\label{eq:p3}
			\end{align}
			Taking $\boldsymbol{u}=\boldsymbol{x}^{\perp}$ completes the proof.
		\end{IEEEproof}
		To prove Theorem \ref{theo:initi}, we need to  upper-bound the term on the right hand-side of \eqref{eq:p3}. Specifically, we upper (lower) bound its numerator (denominator). This is shown, respectively, in the following Lemmata~\ref{lem:l3} and \ref{lem:l4}.
		\begin{lemma}
			\label{lem:l3}
			In the setup of Lemma \ref{lem:l2}, if $\overline{\overline{\mathcal{I}^{c}_{0}}} \geq \mu N$, with $\mu$ sufficiently large, then
			\begin{align}
				\lVert \boldsymbol{S}\boldsymbol{u} \rVert_{2}^{2} \leq (1+\delta-\zeta) \overline{\overline{\mathcal{I}^{c}_{0}}},
			\end{align}
			holds with $\delta,\zeta\in (0,1)$ for almost all $B$-band-limited signals with $B\leq N/2$.
		\end{lemma}
		\begin{IEEEproof}
			As mentioned in Section \ref{sec:approxinit}, the FrFT satisfies $$\boldsymbol{H}=\sum_{n,\alpha=0}^{N-1} \boldsymbol{u}_{n,\alpha}\boldsymbol{u}_{n,\alpha}^{H} = N \boldsymbol{I}.$$ Then, from standard concentration inequality on the sum of positive semi-definite matrices, we have
			\begin{align}
				(1-\delta)\leq\sigma_{min}\left(\frac{1}{N}\boldsymbol{H}\right)&\leq \sigma_{max}\left(\frac{1}{N}\boldsymbol{H}\right)\leq (1+\delta),
				\label{eq:ine1}
			\end{align}
			for any constant $\delta\in (0,1)$, where $\sigma_{max}(\cdot)$ ($\sigma_{min}(\cdot)$) denotes the largest (smallest) singular value. Given that $\boldsymbol{S}$ is a sub-matrix of $\boldsymbol{H}$, it follows from \eqref{eq:ine1} that
			\begin{align}
				\sigma_{max}\left(\frac{1}{\overline{\overline{\mathcal{I}^{c}_{0}}}}\boldsymbol{S} \right) &\leq \sigma_{max}\left(\frac{1}{N\hspace{0.2em}\overline{\overline{\mathcal{I}^{c}_{0}}}}\boldsymbol{H}\right)-\zeta\nonumber\\
				&\leq 1+\delta-\zeta,
				\label{eq:ine2}
			\end{align}
			for some constant $\zeta\in (0,1)$, and any $\delta\in(0,1)$. Thus, for almost all signals from \eqref{eq:ine2}, we have
			\begin{align}
				\lVert \boldsymbol{S}\boldsymbol{u} \rVert_{2}^{2} = \left\lvert \boldsymbol{u}^{H}\boldsymbol{S}^{H}\boldsymbol{S}\boldsymbol{u} \right\rvert \leq (1+\delta-\zeta)\overline{\overline{\mathcal{I}^{c}_{0}}},
			\end{align}
			provided that $\overline{\overline{\mathcal{I}^{c}_{0}}} \geq \mu N$, with $\mu>0$ sufficiently large. Using $\boldsymbol{x}^{\perp}=\boldsymbol{u}$ produces
			\begin{align}
				\lVert \boldsymbol{S}\boldsymbol{x}^{\perp} \rVert_{2}^{2} \leq (1+\delta-\zeta)\overline{\overline{\mathcal{I}^{c}_{0}}}.
				\label{eq:res1}
			\end{align}
		\end{IEEEproof}
		\begin{lemma}
			\label{lem:l4}
			In the setup of Lemma \ref{lem:l2}, $\overline{\overline{\mathcal{I}^{c}_{0}}} \geq \mu N$, with $\mu>0$ sufficiently large, then
			\begin{align}
				\lVert \boldsymbol{S}\boldsymbol{x}\rVert_{2}^{2}\geq (1-\delta) \overline{\overline{\mathcal{I}^{c}_{0}}} ,
				\label{eq:lem4}
			\end{align}
			holds with $\delta \in (0,1)$ for almost all $B$-band-limited signals with $B\leq N/2$.
		\end{lemma}
		\begin{IEEEproof}
			Rewrite the left-hand side in \eqref{eq:lem4} as
			\begin{align}
				\lVert \boldsymbol{S}\boldsymbol{x}\rVert_{2}^{2}=\sum_{(n,\alpha)\in \mathcal{I}^{c}_{0} }\lvert \boldsymbol{u}^{H}_{n,\alpha}\boldsymbol{x}\rvert^{2}.
				\label{eq:lem41}
			\end{align}
			Given that $\boldsymbol{S}$ is a sub-matrix of $\boldsymbol{H}$, it follows from \eqref{eq:ine1} that
			\begin{align}
				\sigma_{min}\left(\frac{1}{\overline{\overline{\mathcal{I}^{c}_{0}}}}\boldsymbol{S} \right) &\geq \sigma_{min}\left(\frac{1}{N\hspace{0.2em}\overline{\overline{\mathcal{I}^{c}_{0}}}}\boldsymbol{H}\right)\nonumber\\
				&\geq 1-\delta,
				\label{eq:ine21}
			\end{align}
			for almost all signals, some constant $\delta\in(0,1)$, and provided $\overline{\overline{\mathcal{I}^{c}_{0}}} \geq \mu N$, with $\mu>0$ sufficiently large. It follows from \eqref{eq:ine21} that 
			\begin{align}
				\sum_{(n,\alpha)\in \mathcal{I}^{c}_{0}} \lvert \boldsymbol{u}^{H}_{n,\alpha}\boldsymbol{x}\rvert^{2} \geq (1-\delta) \overline{\overline{\mathcal{I}^{c}_{0}}},
				\label{eq:lem34}
			\end{align}
			for almost all signals. This completes the proof. 
		\end{IEEEproof}
		
		\subsection{Proof of the theorem}
		Putting together \eqref{eq:cos} and \eqref{eq:lem4}, we obtain
		\begin{align}
			\frac{\lVert \boldsymbol{S}\boldsymbol{u} \rVert_{2}^{2}}{\lVert \boldsymbol{S}\boldsymbol{x}\rVert_{2}^{2}} \leq \frac{1+\delta-\zeta}{1-\delta} \delequal \kappa < 1,
			\label{eq:final}
		\end{align}
		where we take $\delta < \zeta/2$. Combining \eqref{eq:init} and \eqref{eq:final} yields
		\begin{align}
			\lVert \boldsymbol{x}\boldsymbol{x}^{H}-\boldsymbol{x}^{(0)}(\boldsymbol{x}^{(0)})^{H} \rVert_{F}^{2} \leq \kappa.
			\label{eq:final4}
		\end{align}
		for almost all signals provided $\overline{\overline{\mathcal{I}^{c}_{0}}} \geq \mu N$, with $\mu>0$ sufficiently large. 
	\end{lemma}

	\section{Proof of Theorem \ref{theo:guarantee}}
	\label{app:guarantees}
	\subsection{Preliminaries to the proof}
	\begin{lemma}
		\label{lem:card}
		Assume the setup of Theorem \ref{theo:initi}. Then, for any fixed self-adjoint matrix $\boldsymbol{a}\boldsymbol{a}^{H}$, there exists $\boldsymbol{Q}\in \text{Range}(\mathcal{B}^{*})$ such that
		\begin{align}
			\left\lVert \boldsymbol{Q} -\boldsymbol{a}\boldsymbol{a}^{H} \right\rVert_{2} \leq \delta_{1} \left\lVert \boldsymbol{a}\boldsymbol{a}^{H} \right\rVert_{F}
		\end{align}
		for $\delta_{1}<\frac{1}{4\sqrt{2}}$, provided that $\overline{\overline{\mathcal{I}^{c}_{0}}} \geq \mu N$ with $\mu>0$ sufficiently large.
	\end{lemma}
	\begin{IEEEproof}
		Without loss of generality, assume $\lVert \boldsymbol{a}\rVert_{2}=1$ and $\lVert\boldsymbol{u}_{n,\alpha}\rVert_{2}=1$. Then, we set
		\begin{align}
			\boldsymbol{Q} = \frac{1}{N \overline{\overline{\mathcal{I}^{c}_{0}}}}\sum_{(n,\alpha)\in \mathcal{I}^{c}_{0}}\boldsymbol{u}_{n,\alpha}\boldsymbol{u}_{n,\alpha}^{H},
		\end{align}
		which is of the form $\mathcal{B}^{*}(\boldsymbol{\lambda})$. Here, the set $\mathcal{I}^{c}_{0}$ is chosen to have a close estimation of $\boldsymbol{a}$ using Algorithm \ref{alg:initialization} with $\overline{\overline{\mathcal{I}^{c}_{0}}} \geq \mu N$ such that $\mu>0$ sufficiently large. Then, notice that from Theorem \ref{theo:initi} we have
		\begin{align}
			\left\lVert \boldsymbol{Q} - \boldsymbol{a}\boldsymbol{a}^{H}\right\rVert_{2} \leq \delta_{1},
			\label{eq:finalLem2}
		\end{align}
		where $\delta_{1}= \rho$ with $\rho \in (0,1)$. We highlight that the constant $\rho$ depends on the cardinality of $\mathcal{I}_{0}^{c}$ which can be tuned such that $\delta_{1}=\rho<\frac{1}{4\sqrt{2}}$ following the criterion in \eqref{eq:final}. QEF.
	\end{IEEEproof}

	\begin{lemma}
		Assume the setup in Theorem \ref{theo:initi}. Then, for any $\boldsymbol{S}\in \mathcal{T}_{\boldsymbol{x}}$, there exists $\boldsymbol{Q}$ of the form $\boldsymbol{Q}=\mathcal{B}^{*}(\boldsymbol{\lambda})$, $\boldsymbol{\lambda} \in \mathbb{R}^{N^{2}}$, such that
		\begin{align}
			\lVert \boldsymbol{Q} -\boldsymbol{S} \rVert_{2} \leq \delta_{1} \lVert \boldsymbol{S} \rVert_{F},
		\end{align}
		holds for $\delta_{1}\in (0,1)$, provided that $\overline{\overline{\mathcal{I}^{c}_{0}}} \geq \mu N$, with $\mu>0$ sufficiently large. This inequality immediately leads to
		\begin{align}
			\lVert \boldsymbol{Q}_{\mathcal{T}^{\perp}_{\boldsymbol{x}}} \rVert_{2} \leq \delta_{1} \lVert \boldsymbol{S} \rVert_{F}, \hspace{1em} \lVert \boldsymbol{Q}_{\mathcal{T}_{\boldsymbol{x}}} - \boldsymbol{S} \rVert_{F} \leq \delta_{2}\lVert \boldsymbol{S} \rVert_{F},
		\end{align}
		for $\delta_{2}\in (0,1)$.
		\label{theo:certificate}
	\end{lemma}
	\begin{IEEEproof}
		The immediate consequences of Lemma \ref{theo:certificate} hold for the following reasons. Since any matrix in $\mathcal{T}_{\boldsymbol{x}}$ has the rank of at most 2,
		\begin{align}
			\lVert \boldsymbol{Q}_{\mathcal{T}_{\boldsymbol{x}}} - \boldsymbol{S} \rVert_{F} \leq \sqrt{2}\lVert \boldsymbol{Q}_{\mathcal{T}_{\boldsymbol{x}}} - \boldsymbol{S} \rVert_{2} &\leq 2\sqrt{2}\lVert \boldsymbol{Q} - \boldsymbol{S} \rVert_{2} \nonumber\\
			&\leq \rho_{1}\lVert \boldsymbol{S} \rVert_{F},
			\label{eq:finalUni1}
		\end{align}
		where $\rho_{1} = 2\sqrt{2}\delta_{1}<1$ with $\delta_{1}<\frac{1}{4\sqrt{2}}$ (to meet conditions over $\delta_{2}$ in proof of Lemma \ref{lem:card}). In brief, we will show that we choose $\delta_{1}$ satisfying this condition. The second inequality in \eqref{eq:finalUni1} follows from $\lVert \boldsymbol{M}_{\mathcal{T}_{\boldsymbol{x}}} \rVert_{2}\leq 2\lVert \boldsymbol{M} \rVert_{2}$ for any $\boldsymbol{M}$. Next, since $\lVert \boldsymbol{M}_{\mathcal{T}^{\perp}_{\boldsymbol{x}}} \rVert_{2}\leq \lVert \boldsymbol{M} \rVert_{2}$, we get
		\begin{align}
			\lVert \boldsymbol{Z}_{\mathcal{T}^{\perp}_{\boldsymbol{x}}} \rVert_{2} = \lVert \boldsymbol{Z}_{\mathcal{T}^{\perp}_{\boldsymbol{x}}} - \boldsymbol{S}_{\mathcal{T}^{\perp}_{\boldsymbol{x}}} \rVert_{2} \leq \lVert \boldsymbol{Z} - \boldsymbol{S} \rVert_{2} \leq \delta_{1}\lVert \boldsymbol{S} \rVert_{F}.
		\end{align}
		It thus suffices to prove \eqref{eq:finalUni1}. To this end, consider the eigenvalue decomposition of $\boldsymbol{S} = \sigma_{1}\boldsymbol{a}\boldsymbol{a}^{H}+\sigma_{2}\boldsymbol{b}\boldsymbol{b}^{H}$ with normalized vectors $\boldsymbol{a},\boldsymbol{b}\in \mathbb{C}^{N}$. The proof of Lemma \ref{theo:certificate} then follows from Lemma \ref{lem:solution} combined with Lemma~\ref{lem:card}.
	\end{IEEEproof}
	
	\subsection{Proof of the theorem}
	
	Without loss of generality, assume $\lVert \boldsymbol{x} \rVert_{2}=1$.\vspace{0.5em}
	
	\noindent To prove \textbf{A1}: By definition of operator $\mathcal{B}(\cdot)$ in \eqref{eq:auxOperator}, for any $\boldsymbol{S}\in \mathcal{T}_{\boldsymbol{x}}$, we have
	\begin{align}
		\frac{1}{N}\lVert\mathcal{B}(\boldsymbol{S})\rVert_{1} &= \frac{1}{N}\sum_{\alpha,n =0}^{N-1}\left \lvert \text{Tr}\left(\boldsymbol{U}_{n,\alpha}\boldsymbol{S}\right) \right \rvert, \nonumber\\
		&=\frac{1}{N}\sum_{\alpha,n =0}^{N-1}\left \lvert \boldsymbol{u}_{n,\alpha}^{H}\boldsymbol{S}\boldsymbol{u}_{n,\alpha} \right \rvert,
		\label{eq:unique1}
	\end{align}
	Since $\boldsymbol{S}\in \mathcal{T}_{\boldsymbol{x}}$ has a rank at most two, we can choose normalized vectors $\boldsymbol{a},\boldsymbol{b}\in \mathbb{C}^{N}$ such that $\boldsymbol{S} = \sigma_{1}\boldsymbol{a}\boldsymbol{a}^{H}+\sigma_{2}\boldsymbol{b}\boldsymbol{b}^{H}$. Then, from \eqref{eq:unique1}, 
	\begin{align}
		\frac{1}{N}\lVert\mathcal{B}(\boldsymbol{S})\rVert_{1} &\leq \frac{1}{N}\sum_{n,\alpha=0}^{N-1}\lvert \sigma_{1} \rvert \left \lvert \boldsymbol{u}_{n,\alpha}^{H} \boldsymbol{a} \right \rvert^{2} + \lvert \sigma_{2} \rvert \left \lvert \boldsymbol{u}_{n,\alpha}^{H} \boldsymbol{b} \right \rvert^{2} \nonumber\\
		&= \frac{1}{N}\sum_{\alpha=0}^{N-1}\lvert \sigma_{1} \rvert \lVert \boldsymbol{U}_{\alpha}\boldsymbol{a} \rVert_{2}^{2} + \lvert \sigma_{2} \rvert \lVert \boldsymbol{U}_{\alpha}\boldsymbol{b} \rVert_{2}^{2} \nonumber\\
		&\leq \frac{1}{N}\sum_{\alpha=0}^{N-1} (1+\delta)\left( \lvert \sigma_{1} \rvert + \lvert \sigma_{2} \rvert \right) = (1+\delta)\lVert \boldsymbol{S} \rVert_{1},
		\label{eq:unique2}
	\end{align}
	for any $\delta>0$. In the second equality of \eqref{eq:unique2}, the matrix $\boldsymbol{U}_{\alpha}=[\boldsymbol{u}_{0,\alpha},\dots, \boldsymbol{u}_{N-1,\alpha}]^{H}$ (for simplicity, we assume each $\boldsymbol{u}_{n,\alpha}$ is unitary) satisfies $\lVert \boldsymbol{U}_{\alpha} \rVert_{2}\leq 1+\delta$ because it is the FrFT matrix for $\alpha$. The inequality in \eqref{eq:unique2} follows from $\lVert \boldsymbol{S} \rVert_{1}=\lvert \sigma_{1} \rvert + \lvert \sigma_{2} \rvert$.
	
	On the other hand, from \eqref{eq:unique1}, we also conclude that
	\begin{align}
		\frac{1}{N}\lVert\mathcal{B}(\boldsymbol{S})\rVert_{1} &\geq \frac{1}{N}\sum_{\alpha,n =0}^{N-1}  \boldsymbol{u}_{n,\alpha}^{H} \boldsymbol{S}\boldsymbol{u}_{n,\alpha} \nonumber\\
		&=\frac{1}{N}\sum_{\alpha=0}^{N-1}\sigma_{1} \lVert \boldsymbol{U}_{\alpha}\boldsymbol{a} \rVert_{2}^{2} + \sigma_{2} \lVert \boldsymbol{U}_{\alpha}\boldsymbol{b} \rVert_{2}^{2} \nonumber\\
		&\geq (1-\delta)(\sigma_{1} + \sigma_{2})= (1-\delta)\lVert \boldsymbol{S} \rVert_{1},
		\label{eq:unique3}
	\end{align}
	for almost all signals provided $\overline{\overline{\mathcal{I}^{c}_{0}}} \geq \mu N$ with $\mu>0$ sufficiently large in which the last equality comes from $\lvert \sigma_{1}\rvert + \lvert \sigma_{2}\rvert= \sigma_{1} + \sigma_{2}=\lVert \boldsymbol{S}\rVert_{1}$ because $\boldsymbol{S}$ is positive semidefinite. Combining \eqref{eq:unique2} and \eqref{eq:unique3}, the result holds. QEF.\vspace{0.5em}
	
	\noindent To prove \textbf{A2}: We now construct the approximate dual certificate $\boldsymbol{Q}$ obeying the conditions of Theorem \ref{theo:guarantee}. To this end, we employ the \textit{golfing} method presented in \cite{gross2011recovering}. Variations of this technique have also been used in other subsequent works  \cite{candes2011robust,candes2013phaselift,li2013sparse}. The special form that we employ is most closely related to the construction in \cite{li2013sparse,candes2015phase}. 
	
	To build our approximate dual certificate $\boldsymbol{Q}$, we partition the entries of the column-wise version of $\boldsymbol{Y}$, the vector $\boldsymbol{y}$, into $K + 1$ different groups so that, hereafter, the operator $\mathcal{B}_{0}$ corresponds to those measurements from the first $L_{0}$ entries, $\mathcal{B}_{1}$ to those from the next $L_{1}$ ones, and so on. Clearly, $L_{0}+L_{1}+\cdots+L_{K}=N^{2}$. The golfing begins with $\boldsymbol{X}^{(0)}=\frac{2}{L_{0}}\mathcal{P}_{\mathcal{T}_{\boldsymbol{x}}}(\mathcal{B}_{0}^{*}(\boldsymbol{1}))$, where $\mathcal{P}_{\mathcal{T}_{\boldsymbol{x}}}(\cdot)$ represents the orthogonal projector to the set $\mathcal{T}_{\boldsymbol{x}}$. Then, for $t = 1,\dots,K$, we inductively define (a) $\boldsymbol{Z}^{(t)}\in \text{Range}(\mathcal{B}^{*}_{t})$ obeying $\lVert \boldsymbol{Z}^{(t)} - \boldsymbol{X}^{(t-1)}\rVert_{2}\leq \delta_{1}\lVert \boldsymbol{X}^{(t-1)} \rVert_{F}$, 
	and (b) $\boldsymbol{X}^{(t)} = \boldsymbol{X}^{(t-1)} - \mathcal{P}_{\mathcal{T}_{\boldsymbol{x}}}(\boldsymbol{Z}^{(t)})$. In the end, we set
	\begin{align}
		\boldsymbol{Q} = \boldsymbol{Z} - \frac{2}{L_{0}}\mathcal{B}_{0}^{*}(\boldsymbol{1}), \hspace{1em} \boldsymbol{Z} = \sum_{t=1}^{K}\boldsymbol{Z}^{(t)}.
	\end{align}
	Note that Lemma \ref{theo:certificate} asserts that $\boldsymbol{Z}^{(t)}$ exists and, for each $t$, we have
	\begin{align}
		\lVert \boldsymbol{X}^{(t)} \rVert_{F} \leq \delta_{2}\lVert \boldsymbol{X}^{(t-1)} \rVert_{F}, \text{ and } \lVert \boldsymbol{Z}^{(t)}_{\mathcal{T}^{\perp}_{\boldsymbol{x}}} \rVert_{2} \leq \delta_{3}\lVert \boldsymbol{X}^{(t-1)} \rVert_{F}.
		\label{eq:certicate1}
	\end{align}
	
	We now show that the constructed $\boldsymbol{Q}$ satisfies the required assumptions from Lemma \ref{theo:certificate}. First, $\boldsymbol{Q}$ is self-adjoint and of the form $\mathcal{B}^{*}(\boldsymbol{\lambda})$ with $\boldsymbol{\lambda}\in \mathbb{R}^{N^{2}}$. Next, consider
	\begin{align}
		\boldsymbol{Q}_{\mathcal{T}_{\boldsymbol{x}}} = \boldsymbol{Z}_{\mathcal{T}_{\boldsymbol{x}}} - \frac{2}{L_{0}}\mathcal{P}_{\mathcal{T}_{\boldsymbol{x}}}(\mathcal{B}_{0}^{*}(\boldsymbol{1})) &= \sum_{t=1}^{B} \mathcal{P}_{\mathcal{T}_{\boldsymbol{x}}}(\boldsymbol{Z}^{(t)}) - \boldsymbol{X}^{(0)} \nonumber\\
		&=\sum_{t=1}^{K}(\boldsymbol{X}^{(t-1)} - \boldsymbol{X}^{(t)}) - \boldsymbol{X}^{(0)} \nonumber\\
		&= -\boldsymbol{X}^{(K)}.
		\label{eq:certicate2}
	\end{align}
	Then, combining \eqref{eq:certicate1} and \eqref{eq:certicate2}, we have
	\begin{align}
		\lVert \boldsymbol{Q}_{\mathcal{T}_{\boldsymbol{x}}} \rVert_{F} \leq \lVert \boldsymbol{X}^{(K)} \rVert_{F} \leq \delta_{2}^{K} \lVert \boldsymbol{X}^{(0)} \rVert_{F},
		\label{eq:solution1}
	\end{align}
	and
	\begin{align}
		\lVert \boldsymbol{Z}_{\mathcal{T}^{\perp}_{\boldsymbol{x}}} \rVert_{2} \leq \sum_{t=1}^{K} \lVert \boldsymbol{Z}^{(t)}_{\mathcal{T}^{\perp}_{\boldsymbol{x}}} \rVert_{2} &\leq \delta_{3} \sum_{t=1}^{K} \lVert \boldsymbol{X}^{(t-1)} \rVert_{F} \nonumber\\
		&\leq \delta_{3} \sum_{t=1}^{K} \delta_{2}^{K} \lVert \boldsymbol{X}^{(0)} \rVert_{F} \nonumber\\
		&\leq \rho \lVert \boldsymbol{X}^{(0)} \rVert_{F},
		\label{eq:orthoCertificate}
	\end{align}
	for some $\rho \in (0,1)$.
	
	Recall that Lemma \ref{lem:solution} states
	\begin{align}
		\left\lVert\frac{2}{L_{0}}\mathcal{B}_{0}^{*}(\boldsymbol{1}) - 2\boldsymbol{I} \right\rVert_{2} \leq \frac{\epsilon}{4},
	\end{align}
	for some $\epsilon\in (0,1)$. Further, for any matrix $\boldsymbol{S}$, we have $\lVert \boldsymbol{S}_{\mathcal{T}_{\boldsymbol{x}}} \rVert_{2}\leq 2 \lVert \boldsymbol{S} \rVert_{2}$ and $\lVert \boldsymbol{S}_{\mathcal{T}^{\perp}_{\boldsymbol{x}}} \rVert_{2}\leq \lVert \boldsymbol{S} \rVert_{2}$. Hence,
	\begin{align}
		\lVert \boldsymbol{X}^{(0)} -2 \boldsymbol{I}_{\mathcal{T}_{\boldsymbol{x}}} \rVert_{2} \leq \frac{\epsilon}{2}, \hspace{1em}\lVert \boldsymbol{Z}_{\mathcal{T}^{\perp}_{\boldsymbol{x}}} - \boldsymbol{Q}_{\mathcal{T}^{\perp}_{\boldsymbol{x}}} -2 \boldsymbol{I}_{\mathcal{T}^{\perp}_{\boldsymbol{x}}}\rVert_{2}\leq \frac{\epsilon}{4}.
		\label{eq:uniqueF1}
	\end{align}
	Since $\boldsymbol{X}^{(0)}$ has a rank at most $2$, this implies
	\begin{align}
		\lVert \boldsymbol{X}^{(0)} \rVert_{F} \leq \sqrt{2} \lVert \boldsymbol{X}^{(0)} \rVert_{2} &\leq \sqrt{2} \lVert \boldsymbol{X}^{(0)} -2 \boldsymbol{I}_{\mathcal{T}_{\boldsymbol{x}}} \rVert_{2} + 2\sqrt{2} \lVert \boldsymbol{I}_{\mathcal{T}_{\boldsymbol{x}}} \rVert_{2} \nonumber\\
		&\leq \frac{\sqrt{2}}{2} (\epsilon + 4).
		\label{eq:uniqueF2}
	\end{align}
	Substituting the above into \eqref{eq:solution1} yields
	\begin{align}
		\lVert \boldsymbol{Q}_{\mathcal{T}_{\boldsymbol{x}}} \rVert_{F} \leq \frac{\sqrt{2}}{2} (\epsilon + 4) \delta_{2}^{K}.
	\end{align}
	
	Note that one could choose $\delta_{2}<\frac{1}{\sqrt{2}(\epsilon+4)}$ because it directly depends on the cardinality of $\mathcal{I}_{0}^{c}$ (cf. the proof of Lemma \ref{theo:certificate}). Combining \eqref{eq:orthoCertificate}, \eqref{eq:uniqueF1} and \eqref{eq:uniqueF2} results in
	\begin{align}
		\lVert \boldsymbol{Q}_{\mathcal{T}^{\perp}_{\boldsymbol{x}}} + 2\boldsymbol{I}_{\mathcal{T}^{\perp}_{\boldsymbol{x}}} \rVert_{2} &\leq \lVert \boldsymbol{Z}_{\mathcal{T}^{\perp}_{\boldsymbol{x}}} \rVert_{2} + \lVert \boldsymbol{Z}_{\mathcal{T}^{\perp}_{\boldsymbol{x}}} - \boldsymbol{Q}_{\mathcal{T}^{\perp}_{\boldsymbol{x}}} - 2\boldsymbol{I}_{\mathcal{T}^{\perp}_{\boldsymbol{x}}}\rVert_{2} \nonumber\\
		&\leq \rho + \frac{\epsilon}{4} < 1,
	\end{align}
	for any $\epsilon\in (0,1)$, implying that $\boldsymbol{Q}_{\mathcal{T}^{\perp}_{\boldsymbol{x}}} \preceq - \boldsymbol{I}_{\mathcal{T}^{\perp}_{\boldsymbol{x}}}$. This completes the proof.
	
	\bibliographystyle{IEEEtran}
	\bibliography{main}

\end{document}